\documentclass[aps,pra,groupedaddress,superscriptaddress,notitlepage,twocolumn,nofootinbib]{revtex4-2}

\usepackage[latin1]{inputenc}
\usepackage{amsthm}
\usepackage{amssymb}
\usepackage{amsmath}
\usepackage{bbold}
\usepackage{bbm}
\usepackage{nonfloat}
\usepackage[pdftex]{hyperref}
\usepackage{braket}
\usepackage{dsfont}
\usepackage{mathdots}
\usepackage{mathtools}
\usepackage{enumerate}
\usepackage{csquotes}
\usepackage{stmaryrd}
\usepackage[cal=boondox]{mathalfa}
\usepackage{graphicx}
\usepackage{stackengine}
\usepackage{scalerel}
\usepackage{array}
\usepackage{makecell}
\newcolumntype{x}[1]{>{\centering\arraybackslash}p{#1}}
\usepackage{tikz}
\usetikzlibrary{shapes.geometric, arrows, arrows.meta, decorations.pathreplacing, decorations.markings, decorations.pathmorphing, angles, quotes}
\usepackage{booktabs}
\usepackage{xfrac}
\usepackage{siunitx}
\usepackage{centernot}
\usepackage{comment}
\usepackage{marvosym}
\usepackage{chngcntr}
\usepackage{times,txfonts}

\tikzstyle{box} = [rectangle, rounded corners, minimum width=1.618cm, minimum height=1cm,text centered, draw=black]
\tikzstyle{arrow} = [thick,->,>=stealth]

\newtheorem{thm}{Theorem}
\newtheorem*{thm*}{Theorem}
\makeatletter
\newcommand{\setthmtag}[1]{
  \let\oldthethm\thethm
  \renewcommand{\thethm}{#1}
  \g@addto@macro\endthm{
    \addtocounter{thm}{-1}
    \global\let\thethm\oldthethm}
  }
\makeatother

\newtheorem{prop}[thm]{Proposition}
\newtheorem*{prop*}{Proposition}
\newtheorem{lemma}[thm]{Lemma}
\newtheorem*{lemma*}{Lemma}
\newtheorem{cor}[thm]{Corollary}
\newtheorem*{cor*}{Corollary}

\newtheorem*{cj*}{Conjecture}
\newtheorem{Def}[thm]{Definition}
\newtheorem*{Def*}{Definition}

\makeatletter
\def\thmhead@plain#1#2#3{%
  \thmname{#1}\thmnumber{\@ifnotempty{#1}{ }\@upn{#2}}%
  \thmnote{ {\the\thm@notefont#3}}}
\let\thmhead\thmhead@plain
\makeatother

\theoremstyle{definition}

\newcommand{\bb}{\begin{equation}}
\newcommand{\bbb}{\begin{equation*}}
\newcommand{\ee}{\end{equation}}
\newcommand{\eee}{\end{equation*}}

\newcommand\ceil[1]{\left\lceil#1\right\rceil}
\newcommand{\eqt}[1]{\stackrel{\mathclap{\scriptsize \mbox{#1}}}{=}}
\newcommand{\leqt}[1]{\stackrel{\mathclap{\scriptsize \mbox{#1}}}{\leq}}
\newcommand{\geqt}[1]{\stackrel{\mathclap{\scriptsize \mbox{#1}}}{\geq}}
\newcommand{\ketbra}[1]{\ket{#1}\!\!\bra{#1}}

\newcommand{\sumno}{\sum\nolimits}

\newcommand{\G}{\mathrm{\scriptscriptstyle \, G}}

\newcommand{\id}{\mathds{1}}
\newcommand{\R}{\mathds{R}}

\newcommand{\C}{\mathds{C}}

\DeclareMathOperator{\Tr}{Tr}
\DeclareMathOperator{\rk}{rk}

\DeclareMathAlphabet{\pazocal}{OMS}{zplm}{m}{n}

\newcommand{\lsmatrix}{\left(\begin{smallmatrix}}
\newcommand{\rsmatrix}{\end{smallmatrix}\right)}

\stackMath
\newcommand\xxrightarrow[2][]{\mathrel{%
  \setbox2=\hbox{\stackon{\scriptstyle#1}{\scriptstyle#2}}%
  \stackunder[5pt]{%
    \xrightarrow{\makebox[\dimexpr\wd2\relax]{$\scriptstyle#2$}}%
  }{%
   \scriptstyle#1\,%
  }%
}}

\newcommand{\tends}[2]{\xxrightarrow[\! #2 \!]{\mathrm{#1}}}
\newcommand{\tendsn}[1]{\xxrightarrow[\! n\rightarrow \infty\!]{\mathrm{#1}}}

\stackMath

\makeatletter
\newcommand*\rel@kern[1]{\kern#1\dimexpr\macc@kerna}
\newcommand*\widebar[1]{%
  \begingroup
  \def\mathaccent##1##2{%
    \rel@kern{0.8}%
    \overline{\rel@kern{-0.8}\macc@nucleus\rel@kern{0.2}}%
    \rel@kern{-0.2}%
  }%
  \macc@depth\@ne
  \let\math@bgroup\@empty \let\math@egroup\macc@set@skewchar
  \mathsurround\z@ \frozen@everymath{\mathgroup\macc@group\relax}%
  \macc@set@skewchar\relax
  \let\mathaccentV\macc@nested@a
  \macc@nested@a\relax111{#1}%
  \endgroup
}

\newcommand{\gie}{E^\G_\downarrow}
\newcommand{\gietilde}{\widetilde{E}^\G_\downarrow}
\newcommand{\gier}{E^{\G,\infty}_\downarrow}
\newcommand{\reof}{E_{F,2}^\G}
\newcommand{\kg}{K^{\G}}
\newcommand{\kgm}{K^{\G,\mathrm{\scriptscriptstyle\, M}}}
\newcommand{\kp}{K_{\mathcal{P}}}
\newcommand{\M}{\pazocal{M}}
\newcommand{\K}{\pazocal{K}}
\newcommand{\CC}{\pazocal{C}}
\newcommand{\F}{\pazocal{F}}
\newcommand{\Z}{\pazocal{Z}}

\DeclareSymbolFontAlphabet{\mathcalorig}{symbols}

\newcommand{\fakepart}[1]{
 \par\refstepcounter{part}
  \sectionmark{#1}
}

\counterwithin*{equation}{part}
\counterwithin*{thm}{part}
\counterwithin*{figure}{part}

\usepackage{orcidlink}
\begin{document}

\title{Fundamental limitations to key distillation from Gaussian states with Gaussian operations}

\author{Ludovico Lami \orcidlink{0000-0003-3290-3557}}
\email{ludovico.lami@gmail.com}
\affiliation{School of Mathematical Sciences and Centre for the Mathematics
and Theoretical Physics of Quantum Non-Equilibrium Systems (CQNE), University of Nottingham, University Park, Nottingham NG7 2RD, United Kingdom}
\affiliation{Institut f\"{u}r Theoretische Physik und IQST, Universit\"{a}t Ulm, Albert-Einstein-Allee 11, D-89069 Ulm, Germany}
\affiliation{QuSoft, Korteweg-de Vries Institute for Mathematics, and Institute for Theoretical Physics, University of Amsterdam, Science Park, 1098 XG Amsterdam, the Netherlands}

\author{Ladislav Mi\v{s}ta, Jr. \orcidlink{0000-0002-1979-7617}}
\email{mista@optics.upol.cz}
\affiliation{Department of Optics, Palack\'{y} University, 17. listopadu 12, 771 46 Olomouc, Czech Republic}

\author{Gerardo Adesso \orcidlink{0000-0001-7136-3755}}
\email{gerardo.adesso@nottingham.ac.uk}
\affiliation{School of Mathematical Sciences and Centre for the Mathematics
and Theoretical Physics of Quantum Non-Equilibrium Systems (CQNE), University of Nottingham, University Park, Nottingham NG7 2RD, United Kingdom}

\begin{abstract}
We establish fundamental upper bounds on the amount of secret key that can be extracted from quantum Gaussian states by using only local Gaussian operations, local classical processing, and public communication. For one-way public communication, or when two-way public communication is allowed but Alice and Bob first perform destructive local Gaussian measurements, we prove that the key is bounded by the R\'enyi-$2$ Gaussian entanglement of formation $E_{F,2}^{\mathrm{\scriptscriptstyle G}}$. Since the inequality is saturated for pure Gaussian states, this yields an operational interpretation of the R\'enyi-$2$ entropy of entanglement as the secret key rate of pure Gaussian states that is accessible with Gaussian operations and one-way communication. In the general setting of two-way communication and arbitrary interactive protocols, we argue that $2 E_{F,2}^{\mathrm{\scriptscriptstyle G}}$ is still an upper bound on the extractable key. 
We conjecture that the factor of $2$ is 
spurious, which would imply that $E_{F,2}^{\mathrm{\scriptscriptstyle G}}$ coincides with the secret key rate of Gaussian states under Gaussian measurements and two-way public communication. We use these results to prove a gap between the secret key rates obtainable with arbitrary versus Gaussian operations. Such a gap is observed for states produced by sending one half of a two-mode squeezed vacuum through a pure loss channel, in the regime of sufficiently low squeezing or sufficiently high transmissivity. Finally, for a wide class of Gaussian states that includes all two-mode states, we prove a recently proposed conjecture on the equality between $E_{F,2}^{\mathrm{\scriptscriptstyle G}}$ and the Gaussian intrinsic entanglement. The unified entanglement quantifier emerging from such an equality is then endowed with a direct operational interpretation as the value of a quantum teleportation game.
\end{abstract}


\maketitle
\fakepart{Main text}

\section{Introduction}

Quantum entanglement enables distant parties to generate a shared secret key by employing public discussion only~\cite{bennett1984quantum, Ekert91, Bennett1992}, a feat impossible in the classical setting~\cite{Shannon1949, Maurer1993} without additional assumptions on the information available to the eavesdropper~\cite{Wyner1975, CsiszarKoerner, Maurer1993, Maurer1994, AhlswedeCsiszar1}. In the last decades, quantum key distribution (QKD) has established itself as a fundamental primitive in quantum cryptography, thus gaining a central role in the flourishing quantum information science and technology~\cite{QKD-review}. Accordingly, the amount of secret key that can be extracted from a state is regarded as an entanglement measure of fundamental operational importance~\cite{Christandl-Master, squashed,Horodecki05, faithful}. \phantom{$\mathcal{P}$}

Continuous variable (CV) platforms, based on communication over quantum optical modes~\cite{Braunstein-review, HOLEVO-CHANNELS-2}, transmitted either via optical fibres or across free space~\cite{Pirandola2021, Sidhu2021}, have been of paramount importance in the demonstration of QKD. Recently, they witnessed impressive experimental progress~\cite{CanaryIslands-1, CanaryIslands-2, QKD-1200, QKD-intercontinental} and will likely play a major role in any future large-scale technological implementation of QKD~\cite{weedbrook12, CVQKD-review, Laudenbach2018}.

Paradigmatic examples of CV QKD protocols are those based on Gaussian states and Gaussian measurements~\cite{Ralph1999, Hillery2000, Reid2000, Gottesman2001, Cerf01, Grosshans02, Grosshans2003, Weedbrook04, GarciaPatron09, GarciaPatron-Thesis, Tserkis2020, Mountogiannakis2022}. The main advantage of this all-Gaussian paradigm~\cite{weedbrook12, adesso14, BUCCO} is that it is relatively experimentally friendly: coherent states~\cite{Schroedinger1926-coherent, Klauder1960, Glauber1963, Sudarshan1963}, squeezed states~\cite{Kennard1927, Stoler1970, Yuen1976, Slusher1985, Andersen2016, Schnabel2017}, homodyne and heterodyne detection~\cite{Braunstein-review, BUCCO} are nowadays relatively inexpensive ingredients, especially compared to general quantum states and operations. At the same time, it is still quite powerful in the context of QKD: in fact, it has been shown that any sufficiently entangled Gaussian state can be used, in combination with local Gaussian operations and public communication, to distil a secret key~\cite{Navascues2005, Navascues2005b, Rodo2007}. The effectiveness of the all-Gaussian paradigm in QKD is in stark contrast with its fundamentally limited performances at many other important tasks, such as universal quantum computation~\cite{Bartlett2002, Menicucci2006, Ohlinger2010, Mari-Eisert}, entanglement distillation~\cite{nogo1, nogo2, nogo3}, error correction~\cite{Niset2009}, and state transformations in general resource theories~\cite{G-resource-theories, assisted-Ryuji}.

In this paper, we investigate the operational effectiveness of the all-Gaussian framework in the context of QKD, establishing ultimate limitations on the amount of secret key that can be extracted from arbitrary multi-mode Gaussian states by means of local Gaussian operations, local classical processing, and public communication --- a quantity that we call \emph{Gaussian secret key}. The fact that the initial state is Gaussian and that the available \emph{quantum} operations are Gaussian does not mean that the state will be Gaussian at all stages of the protocol, essentially because local classical operations are entirely unrestricted. For example, Alice could decide to apply a random displacement to her system (say, either $+s$ or $-s$, with equal probabilities), making the resulting state non-Gaussian.

In a nutshell we prove that, while key distillation is indeed possible in the Gaussian setting, it is not as efficient as it could be if also non-Gaussian measurements were allowed. Our bounds are given in terms of a Gaussian entanglement measure known as the {\em R\'enyi-$2$ Gaussian entanglement of formation} (denoted $\reof$)~\cite{Wolf03, adesso05, AdessoSerafini, Lami16, LL-log-det, extendibility}, and thus endow this quantity with a sound operational meaning. In this context, after formalizing basic definitions on CV systems (Section~\ref{sec:CV}) and Gaussian key distillation protocols (Section~\ref{sec:protocols}),  we establish three main results in Section~\ref{sec:Main}. 

First, if only one-way public communication is allowed then the Gaussian secret key is at most $\reof$, with the inequality saturated for pure Gaussian states (Theorem~\ref{KG_ub_thm}). Secondly, we argue that the Gaussian secret key is anyway limited by $2\reof$ even in the most general setting where we allow two-way public communication (Theorem~\ref{KW_ub_thm}). Lastly, we show that the upper bound $\reof$ --- without the factor $2$ --- still holds even for two-way public communication, provided that Alice and Bob start the protocol with destructive Gaussian measurements (Theorem~\ref{KGM_ub_thm}).

The R\'enyi-$2$ Gaussian entanglement of formation $\reof$ is a monogamous and additive Gaussian entanglement monotone enjoying a wealth of properties~\cite{AdessoSerafini, Lami16, LL-log-det, extendibility}. Moreover, its computation amounts to a simple single-letter optimisation problem that is analytically solvable for all two-mode mixed states~\cite{Wolf03, adesso05}. Instrumental to our approach is the study of the connection between $\reof$ and another Gaussian entanglement measure known as the {\em Gaussian intrinsic entanglement} (denoted $\gie$)~\cite{GIE, GIE-PRA, GIE-PRA2}. In Section~\ref{sec:Conj} we prove that $\gie\leq\reof$ holds for all multi-mode Gaussian states and, more remarkably, we establish the recently conjectured~\cite{GIE} equality $\gie=\reof$ for the vast class of `normal' Gaussian states, which include in particular all two-mode Gaussian states (Theorem~\ref{equality normal thm}).

In Section~\ref{sec:EX} we explore further applications and interpretations of our results. In particular, in the one-way communication scenario we show  that $\reof$ is often smaller than the one-way distillable entanglement on the physically relevant class of states obtained by sending one half of a two-mode squeezed vacuum across a pure loss channel, entailing that restricting to Gaussian operations leads to a decrease of distillable key. We also provide a general operational interpretation for $\reof$ in a game-theoretical context based on quantum teleportation in the presence of a malicious jammer. We present our concluding remarks in Section~\ref{sec:End}.

\section{Continuous variable basics}\label{sec:CV}
We start by recalling the formalism of CV Gaussian states and measurements~\cite{Wang2007, weedbrook12, BUCCO, adesso14}; see Appendix~\ref{app:Gauss} for further details. 

\subsection{Phase space representations}
For a CV system made of $m$ harmonic oscillators (modes), the displacement operator associated with a vector $\xi\in \R^{2m}$ is defined by~\cite[Section~3.1]{BUCCO}
\bb
D(\xi)\coloneqq e^{i\xi^\intercal \Omega r}\, .
\label{displacement}
\ee
Note that $D(\xi)$ is a unitary operator. Furthermore, for all $\xi$ it holds that $D(\xi)^\dag = D(-\xi)$. The canonical commutation relations can be rewritten in the so-called Weyl form in terms of the displacement operators. They read~\cite[Eq.~(3.11)]{BUCCO}
\bb
D(\xi_1) D(\xi_2) = e^{-\frac{i}{2}\, \xi_1^\intercal \Omega \xi_2} D(\xi_1+\xi_2)\, .
\label{Weyl}
\ee

The characteristic function of an $m$-mode quantum state $\rho$ is the function $\chi_\rho:\R^{2m}\to \C$ defined by~\cite[Section~4.3]{BUCCO}
\bb
\chi_\rho(\xi)\coloneqq \Tr[\rho D(-\xi)]\, .
\label{chi}
\ee
Its Fourier transform is the Wigner function, in formula
\bb
W_\rho(u) \coloneqq \frac{1}{2^m \pi^{2m}}\int d^{2m}\xi\, \chi_\rho(\xi)\, e^{-i \xi^\intercal \Omega u}\, .
\label{Wigner}
\ee

\subsection{Gaussian states}\label{sec:Gauss_subsec}

Let $x_j$ and $p_j$ ($1\leq j\leq m$) be the canonical operators of an $m$-mode CV system, whose vacuum state we denote with $\ket{0}$. Defining the vector $r \coloneqq (x_1,\ldots, x_m, p_1, \ldots, p_m)^\intercal$, the canonical commutation relations can be written in matrix notation as $[r,r^\intercal ] = i \Omega$, where 
\bb
\Omega\coloneqq \begin{pmatrix} 0_m & \id_m \\ -\id_m & 0_m \end{pmatrix}\,.
\ee

A quadratic Hamiltonian is a self-adjoint operator of the form $H_q = \frac12 r^\intercal K r + t^\intercal r$, where $K>0$ is a $2m\times 2m$ real matrix, and $t\in \R^{2m}$. {\it Gaussian states} are by definition thermal states of quadratic Hamiltonians (and limits thereof). They are uniquely defined by their mean or displacement vector, expressed as $s \coloneqq \Tr[\rho\, r ] \in \mathds{R}^{2m}$~\cite{notemean}, and by their quantum covariance matrix (QCM), a $2m\times 2m$ real symmetric matrix given by $V\coloneqq \Tr \left[\rho \left\{r-s,(r-s)^\intercal\right\} \right]$.
Physical QCMs $V$ satisfy the Robertson--Schr\"odinger uncertainty principle \bb V \geq i \Omega, \ee hereafter referred to as {\em  bona fide} condition~\cite{simon94}, which implies $V>0$ and $\det V\geq 1$. Any real matrix $V$ that obeys the bona fide condition is the QCM of some Gaussian state, which is pure iff $\det V=1$. The Gaussian state with mean $s$ and QCM $V$ will be denoted by $\rho_\G[V,s]$. Note that mean vectors and QCMs compose with direct sum under tensor products, $\rho_\G[V,s]\otimes \rho_\G[W,t] = \rho_\G[V\oplus W,\, s\oplus t]$~\cite{notedirect}.

The characteristic function as well as the Wigner function of Gaussian states are in fact Gaussian. More precisely,
\begin{align}
\chi_{\rho_\G[V,s]}(\xi) &= \exp\left[ -\frac14 \xi^\intercal \Omega^\intercal V \Omega \xi + i s^\intercal \Omega \xi \right] , \label{chi_Gaussian} \\
W_{\rho_\G[V,s]}(u) &= \frac{2^m}{\pi^m\sqrt{\det V}} \exp\left[ - (u-s)^\intercal V^{-1} (u-s) \right] . \label{Wigner_Gaussian}
\end{align}
Note that $W_{\rho_\G[V,s]}$ is a Gaussian with mean $s$ and covariance matrix $V/2$. In particular, the differential entropy $H\left( W_{\rho_\G[V,s]}\right) \coloneqq - \int d^{2m} u\, W_{\rho_\G[V,s]}(u) \log_2 W_{\rho_\G[V,s]}(u)$ of the Wigner function associated with a Gaussian state $\rho_\G[V,s]$ evaluates to
\bb
H\left( W_{\rho_\G[V,s]}\right) = \frac12 \log_2 \det V + m \log_2 (\pi e) = M(V) + m \log_2 (\pi e) \, ,
\label{Wigner_entropy_Gaussian}
\ee
where we used the notation
\bb
M(V) \coloneqq \frac12 \log_2 \det V\, .
\label{M}
\ee
If $V$ is a QCM, then $\det V = \prod_{j=1}^m \nu_j^2(V)$ is the squared product of the symplectic eigenvalues. Since these are no smaller than $1$, we conclude that $\det V\geq 1$ and hence $M(V)\geq 0$. Therefore, for all Gaussian states it holds that
\bb
H\left( W_{\rho_\G[V,s]}\right) \geq m \log_2 (\pi e) \, .
\label{lower_bound_M_Gaussian}
\ee

\subsection{Gaussian measurements}\label{sec:GMP}

Quantum measurements are modelled by positive operator-valued measures (POVM) $E(dx)$ over a measure space $\pazocal{X}$, with outcome probability distribution being $p(dx)=\Tr[\rho\, E(dx)]$. 

{\it Gaussian measurements} over an $m$-mode system are  defined by the POVM \bb
E(d^{2m} x) = \rho_G[\Gamma, x]\, \frac{d^{2m}x}{(2\pi)^m}
\label{Gaussian measurement POVM}
\ee
on the measure space $\pazocal{X}=\R^{2m}$, 
with the QCM $\Gamma$ denoting the \emph{seed} of the measurement. We will sometimes represent this as the quantum--classical channel
\bb
\M_{\Gamma}^\G (\cdot) \coloneqq \int \frac{d^{2m}x}{(2\pi)^m}\, \Tr\left[ \rho_G[\Gamma, x] (\cdot) \right] \ketbra{x}\, ,
\label{G_meas_channel}
\ee
where the vectors $\ket{x}$ are formally orthonormal.\nocite{Barchielli2006,Holevo-Kuznetsova}\footnote{For a mathematically rigorous notion of quantum--classical channel, see Barchielli and Lupieri~\cite{Barchielli2006} and also Holevo and Kuznetsova~\cite{Holevo-Kuznetsova}.}

On a Gaussian state $\rho_\G[V, s]$, the Gaussian measurement in~\eqref{Gaussian measurement POVM} yields as outcome a random variable $X\in \R^{2m}$ whose probability density function reads~\cite[Section~5.4.4]{BUCCO}
\bb
p(x) = \frac{e^{-(x-s)^\intercal (V+\Gamma)^{-1}(x-s)}}{\pi^m \sqrt{\det (V+\Gamma)}}\, .
\label{Gaussian measurement PDF}
\ee
In other words, $X$ is normally distributed with mean $s$ and covariance matrix $(V+\Gamma)/2$.\footnote{The factor $1/2$  depends on the different conventions chosen for QCMs and classical covariance matrices. For example, we have defined the first entry of the QCM of an $m$-mode state $\rho$ with vanishing displacement vector to be $V_{11} = 2 \Tr[x_1^2\rho]$, which is twice the variance of the observable $x_1$ on $\rho$.}

If the measured system $A$ is part of a bipartite system $AB$ initially in a Gaussian state $\rho_\G [V_{AB},\, s_{AB}]$, the post-measurement state on $B$ conditioned on obtaining the outcome $x$ is again Gaussian, has QCM given by the {\it Schur complement}
\bb
V'_B = (V_{AB}+\Gamma_A)/(V_A+\Gamma_A)\, ,
\label{post_measurement_QCM}
\ee
and displacement vector that depends on $V_{AB}$, $s_{AB}$, and $x$ as reported in \cite[Section~5.4.5]{BUCCO}. Importantly, note that the post-measurement state of Gaussian measurements depends on the measurement outcome only through its mean vector and not through its QCM: indeed, the expression~\eqref{post_measurement_QCM} is independent of $x$.

From the above interpretation of the expression~\eqref{post_measurement_QCM} as the QCM of the reduced post-measurement state it immediately follows that $V'_B$ is also a QCM, i.e.\ it satisfies 
\begin{equation}
V'_B\geq i\Omega_B\, .
\label{elementary_pure_G_lemma}
\end{equation}
Moreover, one can also conclude that $V'_B$ must be pure if such are both $V_{AB}$ and $\Gamma_A$. These important facts can also be established directly by exploiting the properties of Schur complements reviewed in Appendix~\ref{Schur_subsec}.\footnote{Doing this is a useful exercise that the interested reader is encouraged to solve on their own.}

 The simplest unitary operations one can account for in the Gaussian formalism are so-called {\em Gaussian unitaries}, constructed as products of factors of the form $e^{-iH_q\tau}$, with $H_q$ a quadratic Hamiltonian. For a Gaussian unitary $\pazocal{U}$, the induced state transformation $\rho\mapsto \pazocal{U}\rho \pazocal{U}^\dag$ becomes $V\mapsto SVS^\intercal$ at the level of QCMs. Here $S$ is a $2m\times 2m$ symplectic matrix, satisfying $S\Omega S^\intercal = \Omega$.

A \emph{Gaussian measurement protocol} on a CV system $A$ conditioned on a random variable $U$ is a procedure of the following form: (i)~We append to $A$ a single-mode ancilla $R_1$ in the vacuum state; conditioned on $U$, we apply a Gaussian unitary to $AR_1$ and perform a Gaussian measurement with seed $\Gamma_{1}^U$ on the last $m_1$ modes of the resulting state (possibly, $m_1=0$), obtaining a random variable $X_1\in \R^{2m_1}$. Note that both the Gaussian unitary and $\Gamma_1^U$ may depend on $U$. The modes remaining after the measurement form a system that we denote with $A_1$. (ii)~We append to $A_1$ a single-mode ancilla $R_2$ in the vacuum state, and use $X_1$ together with $U$ to decide on a Gaussian unitary to apply to $A_1 R_2$ and on a Gaussian measurement with seed $\Gamma_2^{X_1U}$ to carry out on the last $m_2$-modes of the resulting state (possibly, $m_2=0$). (iii)~We continue in this way, until after $r$ rounds the protocol terminates. The output products are a random variable $X = (X_1,\ldots, X_r)\in \R^{2\sum_{i=1}^r m_i}$ (the \emph{measurement outcome}) and a quantum system $A_r$.\footnote{Note that allowing general Gaussian states for the ancillae and general non-deterministic Gaussian operations instead of Gaussian unitaries  does not lead to a wider set of protocols. In fact, recall that any Gaussian state can be prepared by applying a Gaussian unitary to the vacuum and discarding some modes, and that non-deterministic Gaussian operations can always be realised by appending ancillae in the vacuum state, acting with Gaussian unitaries, and performing Gaussian measurements on some of the modes~\cite[Sections~5.3--5.5]{BUCCO}.}

\subsection{Gaussian  entanglement measures}\label{sec:GEM106}
In this paper we will relate the {\it secret key} that can be distilled by means of Gaussian protocols to the {\it quantum correlations} contained in Gaussian states.
In general, operationally motivated correlation quantifiers for quantum states are usually based on the von Neumann entropy
\bb
S_1(\rho)\coloneqq -\Tr[\rho \log_2 \rho]\, ,
\label{entropy}
\ee
which is the correct quantum generalisation of the Shannon entropy for classical random variables. Other R\'enyi-$\alpha$ entropies, given for $\alpha\geq 1$ by
\bb
S_\alpha(\rho) \coloneqq \frac{1}{1-\alpha}\log_2 \Tr[\rho^\alpha]\, ,
\label{Renyi_entropy}
\ee
although mathematically important, are commonly thought not to have such a direct operational meaning. However, in the constrained Gaussian setting we study here our interest lies not in the correlations possessed by the state per se, but rather in that part of them that can be accessed by the local parties. Since we also assume that these are restricted to Gaussian measurements, we in fact want to look at the correlations displayed by the classical random variables that constitute the outcomes of those measurements. When the random variable $X$ models the outcome of a Gaussian measurement with seed $\sigma$ performed on the Gaussian state $\rho_\G[V,s]$, its Shannon differential entropy, generally defined by the formula $H(X) = - \int d^{2n}x\, p_X(x) \log_2 p_X(x)$, takes the form (cf.~\eqref{Wigner_entropy_Gaussian})
\bb
H(X) = \frac12 \log_2 \det (V+\sigma) + n \log_2 (\pi e)\, .
\label{differential entropy G outcomes}
\ee
This kind of expression, basically the log-determinant of a QCM, up to additive constants, resembles that appearing in the formula for the R\'enyi-$2$ entropy of the state,
\bb
S_2\left(\rho_\G[V,s]\right) = \frac12 \log_2 \det V = M(V)\, ,
\label{Renyi-2}
\ee
where $M$ has already been defined in~\eqref{M}. Although~\eqref{differential entropy G outcomes} and~\eqref{Renyi-2} are not identical, they share the same functional form. Hence, in some sense it is the R\'enyi-$2$ entropy, and not the von Neumann entropy, that is connected to the Shannon entropy of the experimentally accessible measurement outcomes, when those measurements are also Gaussian. For this precise reason, one can expect the R\'enyi-$2$ entropy to play a role in quantifying those correlations of Gaussian states that can be extracted via Gaussian measurements~\cite{LL-log-det}. In fact, quantifiers based on the R\'enyi-$2$ entropy and their applications have been extensively investigated~\cite{AdessoSerafini, Kor, Simon16, Lami16, GIE, LL-log-det}.

Let us start by introducing a simple correlation quantifier known as the \textit{classical mutual information} of the quantum state $\rho$~\cite{Terhal-Ep, locking}. It is formally given by 
\bb
I^c(A:B)_\rho \coloneqq \sup_{\M_A,\M_B} I(X:Y)\, ,
\ee
where $\M_A,\M_B$ are measurements on $A$ and $B$ with outcomes being the classical random variables $X$ and $Y$. When $\rho_{AB}=\rho_\G[V_{AB}, s_{AB}]$ is Gaussian, and $\M_A,\M_B$ are also restricted to be Gaussian measurements with seeds $\Gamma_A, \Gamma_B$, the maximal mutual information between the local outcomes becomes the \textit{Gaussian mutual information}, given by~\cite{Lada2011}
\begin{align}
I_M^c(A : B)_V &\coloneqq \sup_{\Gamma_A,\Gamma_B} I(X:Y)_Q = \sup_{\Gamma_A,\Gamma_B} I_M(A:B)_{V_{AB}+\Gamma_A\oplus \Gamma_B}\, , \nonumber \\
Q_{XY} &\coloneqq \left( \M_{\Gamma_A}^\G,\M_{\Gamma_B}^\G \right) \left(\rho_\G\left[ V_{AB}, 0\right] \right) ,
\label{IMc}
\end{align}
where the log-determinant mutual information of a bipartite QCM $V_{AB}$ is given by~\cite{LL-log-det}
\begin{align}
I_M(A:B)_V &\coloneqq M(V_A) + M(V_B) - M(V_{AB}) \nonumber \\
&= \frac12 \log_2 \frac{(\det V_A) (\det V_B)}{\det V_{AB}}\, .
\label{IM}
\end{align}
Proving~\eqref{IMc} using~\eqref{Gaussian measurement PDF} and~\eqref{differential entropy G outcomes} is an elementary exercise that is left to the reader. Its solution rests upon the fact that the conditional mutual information is a balanced entropic expression, and hence the `spurious' constant terms in~\eqref{differential entropy G outcomes} cancel out. 

While~\eqref{IMc} is difficult to compute in general, it is known that~\cite{giedkemode, Lada2011, GIE-PRA}
\bb
I_M^c(A:B)_\gamma = \frac12 I_M(A:B)_\gamma = M(\gamma_A)
\label{IMc pure states}
\ee
for all pure bipartite QCMs $\gamma_{AB}$. 
We present a self-contained proof of this fact in Appendix~\ref{sec properties IM}, Lemma~\ref{classical_mutual_info_pure_lemma}. 
Note that the last equality simply follows from the fact that the local reductions of a pure state all have the same R\'enyi entropies.

Moving on from total correlations to entanglement, we can rely on~\eqref{Renyi-2} and~\eqref{M} to form a version of the entanglement of formation called the \textit{\bfseries R\'enyi-2 Gaussian entanglement of formation}~\cite{AdessoSerafini}:
\bb
\reof (V_{AB}) \coloneqq \inf_{\substack{i\Omega_{AB}\leq \gamma_{AB}\leq V_{AB},\\[0.3ex] \text{$\gamma_{AB}$ pure}}} \frac12 \log_2 \det (\gamma_A)\, .
\label{reof}
\ee
This quantity obeys several properties, most notably it is faithful and {\it monogamous}: for all QCMs $V_{ABC}$, it holds that~\cite[Corollary~7]{Lami16}
\bb
\reof(V_{A:BC}) \geq \reof(V_{A:B}) + \reof(V_{A:C})\, ,
\label{reof monogamy}
\ee
where we use colons to signify the partition we are referring to. When combined with the fact that it comes from a convex roof construction, this implies that $\reof$ is also additive~\cite[Corollary~17]{LL-log-det}:
\bb
\reof \left( V_{AB}^{\oplus n} \right) = n\, \reof (V_{AB})\qquad \forall\ n\, .
\label{reof additivity}
\ee
Incidentally, both~\eqref{reof monogamy} and~\eqref{reof additivity} are easy corollaries of the identity~\cite[Theorem~15]{LL-log-det}
\bb
\reof (V_{AB}) =\frac12 \inf_{V_{ABC}} I_M(A:B|C)_V\, ,
\label{reof = sq}
\ee
where
\begin{align}
I_M(A:B|C)_V &\coloneqq M(V_{AC}) + M(V_{BC}) - M(V_C) - M(V_{ABC}) \nonumber \\
&= \frac12 \log_2 \frac{(\det V_{AC})(\det V_{BC})}{(\det V_C)(\det V_{ABC})}
\label{IM_conditional}
\end{align}
is the {\em log-determinant conditional mutual information}~\cite{LL-log-det}, and the infimum ranges over all extensions $V_{ABC}$ of $V_{AB}$, i.e.\ over all QCMs $V_{ABC}$ such that $\Pi_{AB}V_{ABC}\Pi_{AB}^\intercal=V_{AB}$.
Furthermore, the measure \eqref{reof} is known to coincide~\cite{LL-log-det} with a Gaussian version of the {\em squashed entanglement}~\cite{Tucci1999, squashed, Christandl2005, faithful, Takeoka14b}, and it can be analytically computed in a variety of cases of strong physical interest~\cite{Wolf03, adesso05}.

Following an entirely different path, a new entanglement quantifier for Gaussian states has been recently introduced~\cite{GIE, GIE-PRA, GIE-PRA2}. The \textit{\bfseries Gaussian intrinsic entanglement} $\gie(V_{AB})$ of a bipartite Gaussian state with QCM $V_{AB}$ is defined as the minimal \emph{intrinsic information}~\cite{AhlswedeCsiszar1, MaurerWolf, GisinRennerWolf2002, Christandl-Master, RennerWolf, Christandl2003-proceedings, ChristandlRenner, VV2005, Christandl2007} of the classical random variables obtained upon measuring it with Gaussian measurements, assuming that Eve holds a purification of it but her measurement and classical post-processing are also Gaussian. Denoting with $\gamma_{ABE}$ a purification of $V_{AB}$, we get
\bb
\gie( V_{AB}) \coloneqq \sup_{\Gamma_A, \Gamma_B} \inf_{\Gamma_E} I_M\left(A:B|E\right)_{\gamma_{ABE} + \Gamma_A\oplus \Gamma_B\oplus \Gamma_E}\, ,
\label{gie}
\ee
where $\Gamma_A$, $\Gamma_B$, and $\Gamma_E$ are arbitrary QCM on systems $A$, $B$, and $E$, respectively, and $I_M$ is defined in (\ref{IM_conditional}). It is an easy exercise to show that the objective function on the right-hand side of~\eqref{gie} coincides with the conditional mutual information of the triple of random variables generated by carrying out Gaussian measurements with seeds $\Gamma_A, \Gamma_B, \Gamma_E$ on the Gaussian state with QCM $V_{ABE}$. In formula,
\begin{align}
&I_M\left(A:B|E\right)_{\gamma_{ABE} + \Gamma_A\oplus \Gamma_B\oplus \Gamma_E} = I(X:Y|Z)_Q\, , \label{daje} \\
&Q_{XYZ} = \left( \M_{\Gamma_A}^\G \otimes \M_{\Gamma_B}^\G \otimes \M_{\Gamma_E}^\G \right) \left( \rho_\G\left[ V_{ABE},\, 0\right] \right) . \nonumber
\end{align}
The proof of~\eqref{daje} is entirely analogous to that of~\eqref{IMc}. One can show that~\eqref{gie} does not depend on the choice of the purification~$\gamma_{ABE}$ of~$V_{AB}$~\cite{GIE, GIE-PRA, GIE-PRA2}. In order to investigate the asymptotic setting, we will consider the regularisation of~\eqref{gie} as well, given by
\bb
\gier(V_{AB}) \coloneqq \liminf_{n\to\infty} \frac1n \gie\left( V_{AB}^{\oplus n} \right) .
\label{gier_SM}
\ee

It is also worth noticing that~\eqref{IM_conditional} can be cast into the form of an unconditional log-determinant mutual information~\eqref{IM} with the help of Schur complements. Namely, by iteratively applying Schur's determinant factorisation formula~\eqref{det fact} one can easily verify that~\cite[Eq.~(28)]{LL-log-det}
\bb
I_M(A:B|E)_{V_{ABE}} = I_M(A:B)_{V_{ABE}/V_E}\, .
\label{IM_conditional_Schur}
\ee


\section{Gaussian secret key distillation protocols}\label{sec:protocols}

We now consider a communication scenario where two separate parties, Alice and Bob, hold a large number $n$ of copies of a bipartite state $\rho_{AB}$ and want to exploit them to generate a secret key by employing only Gaussian local operations and public communication (GLOPC). It is always understood that we grant them access to local randomness, modelled by random variables that are independent of everything else. 

A generic {\bfseries GLOPC protocol} can be formalised as a quantum-to-classical channel from the bipartite CV system $AB$ to a set of classical alphabets $\K \K' \CC$ (with $\K$ and $\K'$ finite and identical) controlled by Alice, Bob, and the eavesdropper Eve, respectively. Such a protocol will thus be composed of the following steps: (i)~Alice performs a Gaussian measurement protocol on $A$ conditioned on some local random variable $U_1$. (ii)~She uses the measurement outcome $X_1$ together with $U_1$ to prepare a message $C_1$, which is sent to Bob and Eve. (iii)~Bob performs a Gaussian measurement protocol on $B$ conditioned on $C_1$ and on some other local random variable $V_1$. He uses the measurement outcome $Y_1$ together with $V_1$ and $C_1$ to prepare a message $C'_1$, which is sent to Alice and Eve. (iv)~After $2r$ back-and-forth rounds the communication ceases. Alice uses her local random variables $U_1, \ldots, U_r$, the measurement outcomes $X_1, \ldots, X_r$, and Bob's messages $C'_1,\ldots, C'_r$ to prepare a random variable $S$ stored in $\K$, that is, her share of the secret key. Bob does the same with his local random variables $V_1,\ldots, V_r$, his measurement outcomes $Y_1,\ldots, Y_r$, and Alice's messages $C_1,\ldots, C_r$, generating his share of the key $S'$ and storing it into $\K'$. 

In what follows, we will also consider two restricted classes of Gaussian protocols. First, the {\bfseries $\boldsymbol{1}$-GLOPC protocols}, in which public communication is permitted only in one direction, say from Alice to Bob (see Figure~\ref{1GLOPC_fig}). Second, the protocols that can be implemented with Gaussian local (destructive) measurements and public communication ({\bfseries GLMPC protocols}); these start with Alice and Bob making preliminary Gaussian measurements on their entire local subsystems, and then processing only the obtained classical variables with the help of two-way public communication. 

Unless otherwise specified, we will always assume that Eve has access to a purification of the initial quantum state of Alice and Bob and can intercept all publicly exchanged messages (denoted with $C$), storing them in her register $\CC$. 

\begin{figure}[t]
\includegraphics[scale=1]{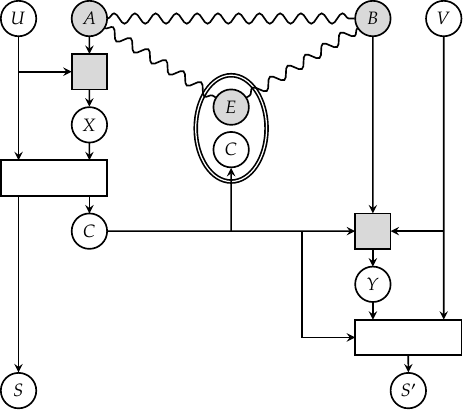}
\caption{A pictorial representation of a generic $1$-GLOPC protocol. (i) Alice performs a Gaussian measurement protocol on $A$ conditioned on some local random variable $U$, obtaining an outcome $X$.  (ii) She uses $X$ together with $U$ to prepare two random variables, $S$ (her share of the secret key) and $C$ (message, sent to Bob and intercepted by Eve).  (iii) Bob performs a Gaussian measurement protocol on $B$ conditioned on $C$ and on some other local random variable $V$, obtaining an outcome $Y$.  (iv) Bob uses $Y$ together with $V$ and $C$ to prepare a random variable $S'$ (his share of the secret key). In the picture, white circles stand for classical random variables, grey circles for quantum systems, and wavy lines for correlations (possibly quantum entanglement). Processes are represented by rectangles, either grey, if they involve quantum systems, or white, if they are purely classical. The central ellipse is Eve's system, which may contain a quantum part $E$ correlated with the initial state of Alice and Bob.}
\label{1GLOPC_fig}
\end{figure}

\subsection{Distillable secret key}
Generally speaking, we say that a number $R>0$ is an \emph{achievable rate} for secret key distillation from the state $\rho_{AB}$ with a class of protocols $\mathcal{P}$ if there exists transformations $\Lambda_n\in \mathcal{P}$ taking as inputs states on $A^nB^n$ and producing as outputs random variables $S_n,S'_n,C_n$ in classical registers $\K_n \K'_n \CC_n$, with the range of $S_n,S'_n$ being $\left\{1,\ldots, 2^{\ceil{Rn}}\right\}$, in such a way that~\cite{Christandl2007}
\bb
\lim_{n\to\infty}\ \inf_{\substack{\Lambda_n: A^n B^n \to \K_n \K'_n \CC_n,\\ \omega_{\CC_n E^n}}} \left\| \Lambda_n \left( \Psi_{ABE}^{\otimes n} \right) - \left(\kappa_{2^{\ceil{Rn}}} \right)_{\K_n \K'_n} \otimes \omega_{\CC_n E^n} \right\|_1 = 0\, .
\label{achievable rate}
\ee
Here, $\Psi_{ABE}$ denotes a purification of $\rho_{AB}$, $\kappa_N\coloneqq \sum_{i=1}^N \ketbra{ii}$ is an ideal secret key of length $\ceil{Rn}$, and the minimisation is over all classical-quantum states $\omega_{\CC_n E^n}$. The meaning of~\eqref{achievable rate} is that the key held by Alice and Bob is asymptotically decoupled from Eve's system, and is thus sufficiently secure to be used in applications~\cite{Koenig2007}.

\begin{Def} \label{SK_rate_def}
The \textbf{$\boldsymbol{\mathcal{P}}$-distillable secret key} $\kp$ of the state $\rho_{AB}$ is the supremum of all rates achievable with protocols in $\mathcal{P}$, 
$\kp (\rho_{AB}) \coloneqq \sup R>0$ such that Eq.~(\ref{achievable rate}) holds.
\end{Def}
In this paper we are naturally interested in the case where $\rho_{AB}$ is a Gaussian state with QCM $V_{AB}$, and the considered protocols are either GLOPC, or $1$-GLOPC, or GLMPC. The associated secret keys are easily seen to depend on $V_{AB}$ only; we will denote them with the shorthand notation $\kg_\leftrightarrow(V_{AB})$, $\kg_\to (V_{AB})$, and $\kgm_\leftrightarrow(V_{AB})$, respectively.

A particularly useful upper bound on $\kp (V_{AB})$ can be established by forcing Eve to apply a Gaussian measurement with pure seed of her choice before the beginning of the protocol, and to broadcast the obtained outcome to Alice and Bob together with the description of $\Gamma_E$. 
From \eqref{post_measurement_QCM} we know that in this case the state that Alice and Bob share is Gaussian and has QCM 
\bb
\tau_{AB} = (\gamma_{ABE}+\Gamma_E)/(\gamma_E+\Gamma_E)\, ,
\label{protocol_eq0}
\ee
where $\gamma_{ABE}$ is a purification of the QCM $V_{AB}$. Since Alice and Bob also know Eve's measurement outcome, they can easily apply local displacements and have their state's mean vanish. The protocols then proceed as detailed earlier in this Section.
Note that Eve's measurement outcome is now independent of Alice and Bob's state and is thus useless. We summarise this discussion by giving the following definition.

\begin{Def} \label{modified_SK_rate_def}
For $\mathcal{P}\in \left\{ \text{\emph{GLOPC, 1-GLOPC, GLMPC}} \right\}$, the \textbf{modified $\boldsymbol{\mathcal{P}}$-distillable secret key} associated with a Gaussian state with QCM $V_{AB}$, denoted $\widetilde{K}_{\mathcal{P}} (V_{AB})$ --- or more succinctly $\widetilde{K}_\leftrightarrow^\G(V_{AB}),\, \widetilde{K}^\G_\to (V_{AB}),\, \widetilde{K}^{\G,\mathrm{\scriptscriptstyle\, M}}_\leftrightarrow(V_{AB})$ --- is the supremum of all numbers $R>0$ such that
\begin{align}
\lim_{n\to\infty} \sup_{\Gamma_{E^n}} \inf_{\substack{ \Lambda_n: A^n B^n \to \K_n \K'_n \CC_n,\\ Q_{\CC_n}}} &\Big\| \Lambda_n \left( \rho_\G\left[ (\gamma_{ABE}^{\oplus n}+\Gamma_{E^n})/(\gamma_E^{\oplus n}+\Gamma_{E^n}),\, 0\right] \right) \nonumber\\ 
&\  - \left( \kappa_{2^{\ceil{Rn}}} \right)_{\K_n\K'_n} \otimes Q_{\CC_n} \Big\|_1 = 0\, ,
\label{modified_SK_rate}
\end{align}
where 
$Q_{\CC_n}$ is an arbitrary probability distribution over the alphabet $\CC_n$.
\end{Def}

It should be clear that the new class of protocols in Definition~\ref{modified_SK_rate_def} allows for a secret key distillation rate that is never smaller than that corresponding to the protocols in Definition~\ref{SK_rate_def}, because in the former case Eve is forced to lose access to her quantum system at an early stage. 
We give a formal proof of this below.

\begin{lemma} \label{K_tilde_ub_lemma}
For all $\mathcal{P}\in \left\{ \text{\emph{GLOPC, 1-GLOPC, GLMPC}} \right\}$ and all QCMs $V_{AB}$, it holds that 
\bb
\kp (V_{AB}) \leq \widetilde{K}_{\mathcal{P}}(V_{AB})\, ,
\label{K_tilde_ub}
\ee
where $\kp (V_{AB})$ and $\widetilde{K}_{\mathcal{P}} (V_{AB})$ are given in Definitions~\ref{SK_rate_def} and~\ref{modified_SK_rate_def}, respectively.
\end{lemma}

\begin{proof}
Let $R>0$ be an achievable rate for $\kp (V_{AB})$. Construct a sequence of protocols $\Lambda_n: A^n B^n \to \K_n \K'_n \CC_n$ of class $\mathcal{P}$, where $\K_n,\K'_n$ are two copies of an alphabet of size $2^{\ceil{Rn}}$, and a sequence of states $\omega_{\CC_n E^n}$, such that
\bb
\lim_{n\to\infty}\ \left\| \Lambda_n \left( \rho_{ABE}^{\otimes n} \right) - \left(\kappa_{2^{\ceil{Rn}}} \right)_{\K_n \K'_n} \otimes \omega_{\CC_n E^n} \right\|_1 = 0\, .
\label{K_tilde_ub_eq0}
\ee
For a fixed $n$, consider an arbitrary QCM $\Gamma_{E^n}$. For a vector $x\in \F_n \coloneqq \R^{2nm_E}$, with $m_E$ being the number of modes of $E$, let $p(x)$ denote the value on $x$ of the probability density function of the outcome of the Gaussian measurement with seed $\Gamma_{E^n}$ on Eve's share of the state $\rho_{ABE}^{\otimes n}$. Also, let $t_{A^nB^n}^x = (t_{A^n}^x,\, t_{B^n}^x)\in \R^{2n(m_A+m_B)}$ be the displacement vector of the post-measurement state on $A^n B^n$ corresponding to the outcome $x$.

We now construct a modified protocol $\Lambda'_n: A^n B^n\to \K_n \K'_n \CC_n \F_n$ of class $\mathcal{P}$, where the measurable space $\F_n = \R^{2n m_E}$ pertains to Eve. To do this, we distinguish two separate cases. If $\mathcal{P}\in \left\{\text{GLOPC, 1-GLOPC}\right\}$, then $\Lambda'_n$ proceeds as follows:
(i) Alice draws a local random variable $X$ on $\F_n$ distributed according to $p$, applies to $A^n$ the displacement unitary $D\left(t_{A^n}^X\right)$, and then continues with her (first) Gaussian measurement protocol as prescribed by $\Lambda_n$. (ii)
 During the (first) round of communication, Alice sends to Bob and Eve not only the message originally prescribed by $\Lambda_n$, but also the random variable $X$.
(iii) Before continuing with his (first) Gaussian measurement protocol dictated by $\Lambda_n$, Bob applies a displacement unitary $D\left(t_{B^n}^X\right)$ to his share of the system.
(iv)  The protocol continues with further communication rounds (if $\mathcal{P}=\text{GLOPC}$) or directly with key generation (if $\mathcal{P}=\text{1-GLOPC}$) as prescribed by $\Lambda_n$.

If instead $\mathcal{P} = \text{GLMPC}$, the modified protocol $\Lambda'_n$ is even simpler:
(i) Alice and Bob apply global Gaussian measurements to their entire subsystems as dictated by $\Lambda_n$, obtaining measurement outcomes $Z_n$ and $W_n$, respectively.
(ii) Before preparing her first message for Bob, Alice draws a local random variable $X$ on $\F_n$ distributed according to $p$ and translates $Z_n$ by $t_{A^n}^X$.
(iii) Alice then sends to Bob not only the message originally prescribed by $\Lambda_n$, but also $X$.
(iv) Before preparing his first message for Alice, Bob translates $W_n$ by $t_{B^n}^X$.
(v) The protocol continues with further communication rounds and then with key generation as prescribed by $\Lambda_n$.

It is not too difficult to verify that in all three cases
\bb
\begin{aligned}
\Lambda'_n(\cdot) =&\ \int d^{2nm_E}x\ p(x)\, \left(\Lambda_n\circ \mathcal{D}(t_{A^nB^n}^x)\right) \left( \cdot \right) \otimes \ketbra{x}_{\F_n}\, , \\
\mathcal{D}(t_{A^nB^n}^x) (\cdot) \coloneqq&\ D(t_{A^n}^x) \otimes D(t_{B^n}^x)\, (\cdot)\, D(-t_{A^n}^x) \otimes D(-t_{B^n}^x)\, ,
\end{aligned}
\label{K_tilde_ub_eq1}
\ee
with the system $\F_n$ storing $X$ being on Eve's side.
Let us now estimate the figure of merit in~\eqref{modified_SK_rate} for this protocol. We have that
\begin{align*}
&\inf_{Q_{\CC_n \F_n}} \big\| \Lambda'_n \left( \rho_\G\left[ \left(\gamma_{ABE}^{\oplus n}\!+\!\Gamma_{E^n}\right)\Big/\left(\gamma_E^{\oplus n}\!+\!\Gamma_{E^n}\right),\, 0\right] \right) - \kappa_{2^{\ceil{Rn}}} \otimes Q_{\CC_n\F_n} \big\|_1 \\
&\eqt{1} \inf_{Q_{\CC_n \F_n}} \bigg\| \int d^{2nm_E}x\ p(x) \left(\Lambda_n\!\circ\! \mathcal{D}(t_{A^nB^n}^x)\right) \\ & \qquad\qquad \times \left( \rho_\G\left[ \left(\gamma_{ABE}^{\oplus n}\!+\!\Gamma_{E^n}\right)\Big/\left(\gamma_E^{\oplus n}\!+\!\Gamma_{E^n}\right),\, 0\right] \right) \\ & \qquad\qquad \otimes \ketbra{x}_{\F_n}  - \kappa_{2^{\ceil{Rn}}} \otimes Q_{\CC_n \F_n} \bigg\|_1  \\
&\leqt{2} \bigg\| \int d^{2nm_E}x\, \Big( p(x)\, \Lambda_n \left( \rho_\G\left[ \left(\gamma_{ABE}^{\oplus n}\!+\!\Gamma_{E^n}\right)\Big/\left(\gamma_E^{\oplus n}\!+\!\Gamma_{E^n}\right),\, t_{A^nB^n}^x \right] \right) \\ &\qquad \otimes \ketbra{x}_{\F_n} - \kappa_{2^{\ceil{Rn}}} \otimes \Tr_{E^n}\! \left[ \omega_{\CC_n E^n} \rho_\G\!\left[ \Gamma_{\!E^n}, x\right] \right] \otimes \ketbra{x}_{\F_n} \Big) \bigg\|_1 \\
&\leqt{3} \big\| \Lambda_n \left( \rho_\G\left[ \gamma_{ABE}^{\oplus n},\, 0 \right] \right) - \kappa_{2^{\ceil{Rn}}} \otimes \omega_{\CC_n E^n} \big\|_1.
\end{align*}
Here: in~1 we used~\eqref{K_tilde_ub_eq1}; in~2 we let the displacement act on the Gaussian state and considered the ansatz
$Q_{\CC_n\F_n} = \int d^{2nm_E}x\ \Tr_{E^n}\! \left[ \omega_{\CC_n E^n} \rho_\G\!\left[ \Gamma_{\!E^n}, x\right] \right] \otimes \ketbra{x}_{\F_n}$, which is nothing but the probability distribution obtained by making the Gaussian measurement with seed $\Gamma_{E^n}$ on $\omega_{\CC_n\F_n}$; finally, 3~follows from the data processing inequality for the trace norm. Taking the supremum over $\Gamma_{E^n}$  yields
$\sup_{\Gamma_{E^n}} \inf_{Q_{\CC_n \F_n}} \Big\| \Lambda'_n \left( \rho_\G\left[ \left(\gamma_{ABE}^{\oplus n}\!+\!\Gamma_{E^n}\right)\Big/\left(\gamma_E^{\oplus n}\!+\!\Gamma_{E^n}\right),\, 0\right] \right) - \kappa_{2^{\ceil{Rn}}} \otimes Q_{\CC_n\F_n} \Big\|_1 
 \leq \left\| \Lambda_n \left( \rho_\G\left[ \gamma_{ABE}^{\oplus n},\, 0 \right] \right) - \kappa_{2^{\ceil{Rn}}} \otimes \omega_{\CC_n E^n} \right\|_1 \tendsn{} 0$,
where  we used~\eqref{K_tilde_ub_eq0}. In light of Definition~\ref{modified_SK_rate_def}, this shows that $R$ is also an achievable rate for $\widetilde{K}_{\mathcal{P}}(V_{AB})$,  concluding the proof.
\end{proof}

\section{Bounds to Gaussian secret key distillation}\label{sec:Main}

We now present our main results establishing fundamental upper bounds on the secret key that can be distilled by means of the Gaussian protocols introduced in Section~\ref{sec:protocols}. To keep the presentation accessible, some auxiliary results and more technical  derivations will be deferred to Appendices.

\subsection{One-way public communication} \label{sec:Tm1}

As announced, our first result is a bound on the 1-GLOPC distillable secret key of an arbitrary Gaussian state.

\begin{thm} \label{KG_ub_thm}
For all QCMs $V_{AB}$, it holds that
\bb
\kg_\to (V_{AB}) \leq \reof(V_{AB})\, .
\label{KG_ub}
\ee
If $V_{AB}=\gamma_{AB}$ is pure, then~\eqref{KG_ub} is tight, i.e.
\bb
\kg_\to(\gamma_{AB}) = \reof(\gamma_{AB}) = \frac12 \log_2 \det (\gamma_A)\, .
\label{reof=KG_pure}
\ee
\end{thm}

Since the right-hand side of~\eqref{KG_ub} does not depend on the direction of communication,~\eqref{KG_ub} holds irrespectively of whether we consider Alice-to-Bob or Bob-to-Alice public communication, as long as we do not allow both. The protocol achieving~\eqref{reof=KG_pure} consists in the application of local homodyne measurements followed by a classical secret key distillation protocol~\cite{Maurer1994, AhlswedeCsiszar1}.
To prove \eqref{KG_ub}, we will make use of the modified secret key $\widetilde{K}^\G_\to(V_{AB})$ introduced in  Definition~\ref{modified_SK_rate_def} and establish the following chain of inequalities, $\kg_\to (V_{AB}) \leq \widetilde{K}^\G_\to(V_{AB}) \leq \reof(V_{AB})$,
where the leftmost one follows from Lemma~\ref{K_tilde_ub_lemma}.

\begin{figure}[t]
\includegraphics[scale=1]{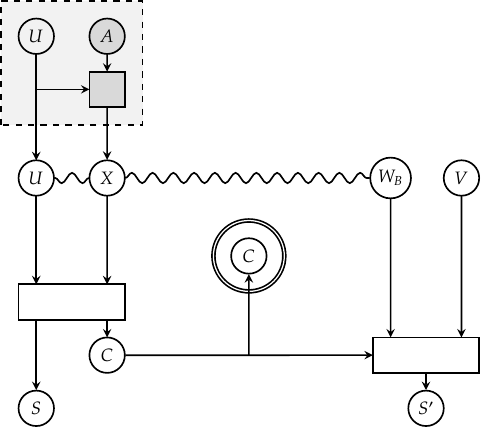}
\caption{The modified protocol used in the proof of Theorem~\ref{KG_ub_thm}. Once the Gaussian protocol on $A$ has been carried out, we can think of the $B$ system as formally simulated by the classical random variable $W_B$ whose distribution conditioned on $UX$ coincides with the Wigner function of the (Gaussian) reduced state on $B$ corresponding to the recorded values of $U$ and $X$.}
\label{1GLOPC_modified_fig}
\end{figure}

\begin{proof}
Let $\gamma_{ABE}$ be a purification of $V_{AB}$, and let us consider a 1-GLOPC protocol applied on the corresponding Gaussian state.
Let us look at the situation right before Bob's measurement (see Figure~\ref{1GLOPC_fig}). Almost all `quantumness' has disappeared, in the sense that the only party still holding a quantum state is Bob. From the point of view of Alice, who knows the value of $U$, the Gaussian measurement protocol she has applied in the first step, and the associated outcome $X$, Bob's state $\rho_{B|U,X}$ is Gaussian. 

We now claim that this situation can be simulated by an entirely classical system. Namely, let $W_B$ be a random variable on the phase space $\R^{2m_B}$ of Bob's system whose probability distribution conditioned on the values of $U$ and $X$ coincides with the Wigner function of $\rho_{B|U,X}$, which is everywhere positive because $\rho_{B|U,X}$ is Gaussian. Noting that: (a)~the vacuum itself has positive Wigner function; (b)~any Gaussian unitary amounts to a linear transformation at the phase space level, and thus preserves the positivity of the Wigner function; and (c)~the POVM elements describing Gaussian measurements also have positive Wigner function, by inspecting the definition of Gaussian measurement protocol (Section~\ref{sec:GMP}) we see that step~(iii) in Figure~\ref{1GLOPC_fig} can be simulated by purely classical operations on $W_B$, $C$, and $V$. We are therefore in the situation depicted in Figure~\ref{1GLOPC_modified_fig}.

We can proceed by following to a certain extent the technique introduced by Maurer~\cite{Maurer1993}. In what follows, we will compute conditional entropies and mutual informations between random variables that are both discrete ($U$, $S$, $C$, $S'$, $V$) and continuous ($X$ and $W_B$). In the latter case it is understood that we employ the differential entropy (measured in bits) instead of the discrete one, although we will denote both with the symbol $H$ for simplicity. Remember that a linear entropy inequality involving differential entropies is valid if and only if its discrete counterpart is `balanced' and valid~\cite{Chan-balanced}. Our derivation rests only upon balanced inequalities. We start by writing
\bb
H(S) = I(S:C) + H(S|C)\, .
\label{protocol_eq1}
\ee
Now,
\bb
\begin{aligned}
H(S|C) &= H(SUX|C) - H(UX|SC) \\
&= H(UX|C) - H(UX|SC) \\
&\leqt{1} H(UX|C) - H(UX|SCVW_B) \\
&= H(UX|C) - H(UXS|CVW_B) + H(S|CVW_B) \\
&\eqt{2} H(UX|C) - H(UX|CVW_B) + H(S|CVW_B) \\
&= I(UX:VW_B|C) + H(S|CVW_B) \\
&\leqt{3} I(UX:VW_B|C) + H(S|S') \\
&\eqt{4} I(UX:W_B|C) + H(S|S')\, .
\end{aligned}
\label{protocol_eq2}
\ee
Here: 1~comes from data processing; 2~is a consequence of the fact that $S$ is a deterministic function of $U$ and $X$; 3~is again data processing, using the fact that $S'$ is a deterministic function of $C$, $V$, and $W_B$; finally, 4~uses that $V$ is independent of $U$, $X$, and $W_B$, even conditioned on $C$. Now, observe that
\bb
\begin{aligned}
I(UX:W_B|C) &= H(W_B|C) - H(W_B|UXC) \\
&= H(W_B|C) - H(W_B|UX) \\
&\leqt{5} H(W_B) - H(W_B|UX) \\
&\leqt{6} H(W_B) - m_B \log_2 (\pi e) \\
&\eqt{7} M(\tau_B)\, .
\end{aligned}
\label{protocol_eq3}
\ee
Note that 5~is just the positivity of the mutual information $I(W_B:C)\geq 0$, 6~is a rephrase of~\eqref{lower_bound_M_Gaussian}, and 7~comes from~\eqref{Wigner_entropy_Gaussian}. Combining~\eqref{protocol_eq1}--\eqref{protocol_eq3} yields
\bb
H(S) \leq I(S:C) + H(S|S') + M(\tau_B) \, .
\label{protocol_eq4}
\ee


We now consider a sequence of $1$-GLOPC protocols $\Lambda_n: A^n B^n \to \K_n \K'_n \CC_n$ with rate $R$ as in Definition~\ref{modified_SK_rate_def}. Pick numbers $\epsilon_n>0$ such that
\begin{align}
\epsilon_n > &\sup_{\Gamma_{E^n}}\ \inf_{Q_{\CC_n}} \Big\| \Lambda_n \left( \rho_\G\left[ \left(\gamma_{ABE}^{\oplus n}+\Gamma_{E^n}\right)\Big/ \left(\gamma_E^{\oplus n}+\Gamma_{E^n}\right),\, 0\right] \right) \nonumber \\
&- \left( \kappa_{2^{\ceil{Rn}}} \right)_{\K_n \K'_n} \otimes Q_{\CC_n} \Big\|_1
\label{protocol_eq5}
\end{align}
and
\bb
\lim_{n\to\infty} \epsilon_n =0\, .
\label{protocol_eq6}
\ee
Now, consider a fixed sequence of QCMs $\Gamma_{E^n}$. Set
\bb
\tau_{A^n B^n} \coloneqq \left(\gamma_{ABE}^{\oplus n}+\Gamma_{E^n}\right)\Big/\left(\gamma_E^{\oplus n}+\Gamma_{E^n}\right) ,
\label{protocol_eq7}
\ee
denote by $S_n,S'_n$ the keys produced by the protocol $\Lambda_n$ run with input $\rho_\G\left[  \tau_{A^nB^n},\, 0\right]$, and let $C_n$ be the message exchanged. Applying~\eqref{protocol_eq4}, we see that
\bb
H(S_n) \leq I(S_n:C_n) + H(S_n|S'_n) + M\left(\tau_{B^n}\right) .
\label{protocol_eq8}
\ee

By tracing away $E_n$, from~\eqref{protocol_eq5} we deduce that the probability distribution $P_{S_n S'_n}$ is at least $\epsilon_n$-close in total variation norm to that of two perfectly correlated copies of the key. In turn, this ensures that $\Pr\{S_n\neq S'_n\} < \epsilon_n$. Hence, Fano's inequality~\cite{FANO} gives
\bb
H(S_n|S'_n) < h_2(\epsilon_n) + \epsilon_n \log_2 \left(|S_n|-1\right) \leq h_2(\epsilon_n) + \ceil{R n} \epsilon_n ,
\label{protocol_eq9}
\ee
where $h_2(x)\coloneqq -x\log_2 x - (1-x) \log_2 (1-x)$ is the binary entropy, and we remembered that $|S_n|=2^{\ceil{Rn}}$.

The same reasoning guarantees that $P_{S_n}$ is at least $\epsilon_n$-close in total variation norm to the uniform distribution over an alphabet of size $2^{\ceil{Rn}}$, whose entropy (measured in bits, as usual) is naturally given by $\ceil{Rn}$. The Fannes--Auedenaert inequality~\cite{Fannes1973, Audenaert2007} thus guarantees that
\begin{align}
H(S_n) &\geq \ceil{Rn} - \frac{\epsilon_n}{2} \log_2 \left( 2^{\ceil{Rn}} - 1\right) - h_2\left(\frac{\epsilon_n}{2}\right) \nonumber \\
&\geq \ceil{Rn} - \frac{\epsilon_n}{2} \ceil{Rn} - h_2\left(\frac{\epsilon_n}{2}\right) .
\label{protocol_eq10}
\end{align}

The last consequence of~\eqref{protocol_eq5} we are interested in can be deduced by tracing away the $K'$ system. By doing so we see that the joint random variable $S_n C_n$ is at least $\epsilon_n$-close in total variation norm to a pair $\widetilde{S}_n \widetilde{C}_n$ of independent random variables such that $\widetilde{S}_n$ is uniformly distributed over an alphabet of size $2^{\ceil{Rn}}$. We deduce that
\bb
\begin{aligned}
I(S_n:C_n) &= H(S_n) - H(S_n|C_n) \\
&\leqt{8} H\big(\widetilde{S}_n\big) - H(S_n|C_n) \\
&\eqt{9} H\big(\widetilde{S}_n \big| \widetilde{C}_n\big) - H(S_n|C_n) \\
&\leqt{10} \frac{\epsilon_n}{2} \log_2 \left(2^{\ceil{Rn}}-1\right) + h_2\left(\frac{\epsilon_n}{2}\right) \\
&\leq \frac{\epsilon_n}{2} \ceil{Rn} + h_2\left(\frac{\epsilon_n}{2}\right) .
\end{aligned}
\label{protocol_eq11}
\ee
The above derivation is justified as follows: in~8 we observed that $H\big(\widetilde{S}_n\big)=\ceil{Rn} \geq H(S_n)$; in~9 we used the fact that $\widetilde{S}_n$ and $\widetilde{C}_n$ are independent; finally, in~10 we exploited the asymptotic continuity of the conditional entropy~\cite{Alicki-Fannes, tightuniform, Alhejji2019, Mark2020}.

Combining~\eqref{protocol_eq7}--\eqref{protocol_eq11} yields the bound
\begin{align}
\ceil{Rn} < &2\epsilon_n \ceil{Rn} + 2\, h_2\left( \frac{\epsilon_n}{2} \right) + h_2(\epsilon_n) \nonumber \\
&+ M\left( \left(\gamma_{AE}^{\oplus n}+\Gamma_{E^n}\right)\Big/\left(\gamma_E^{\oplus n}+\Gamma_{E^n}\right) \right) .
\label{protocol_eq12}
\end{align}
Since this holds for all pure QCMs $\Gamma_{E^n}$, we can take the infimum of the last addend over $\Gamma_{E^n}$. Note that
\begin{align}
\inf_{\text{$\Gamma_{E^n}$ pure QCM}} &M\left( \left(\gamma_{AE}^{\oplus n}+\Gamma_{E^n}\right)\Big/\left(\gamma_E^{\oplus n}+\Gamma_{E^n}\right) \right) \nonumber
\\ &\eqt{11} \frac12\, \inf_{\text{$\Gamma_{E^n}$ pure QCM}} I_M\left(A^n:B^n\right)_{\left(\gamma_{ABE}^{\oplus n}+\Gamma_{E^n}\right)\big/\left(\gamma_E^{\oplus n}+\Gamma_{E^n}\right)} \nonumber \\
&\eqt{12} \reof\left(V_{AB}^{\oplus n}\right) \nonumber \\
&\eqt{13} n\, \reof\left(V_{AB}\right) .
\label{protocol_eq13}
\end{align}
Here, 11~is a consequence of~\eqref{IM pure}, while 12~is an application of the non-trivial fact that the R\'enyi-$2$ Gaussian entanglement of formation coincides with the R\'enyi-2 Gaussian squashed entanglement for all Gaussian states~\cite[Theorem~5 and Remark~2]{LL-log-det} (cf.~\eqref{reof = sq}; remember that $\gamma_{ABE}$ is a purification of $V_{AB}$). Finally, 13~follows from the additivity of the R\'enyi-2 Gaussian entanglement of formation~\cite[Corollary~1]{LL-log-det}. Therefore, optimising~\eqref{protocol_eq12} over pure QCMs $\Gamma_{E^n}$ and using~\eqref{protocol_eq13} yields
\bb
\ceil{Rn} < 2\epsilon_n \ceil{Rn} + 2\, h_2\left( \frac{\epsilon_n}{2} \right) + h_2(\epsilon_n) + n\, \reof\left(V_{AB}\right) .
\label{protocol_eq14}
\ee
Dividing by $n$, taking the limit $n\to\infty$ and using the continuity of the binary entropy together with the fact that $\epsilon_n\tendsn{}0$ finally gives that
\bb
R < \reof(V_{AB})\, .
\label{protocol_eq15}
\ee
Taking the supremum over achievable rates $R$, we then see that
\bb
\widetilde{K}^\G_\to \left( V_{AB} \right) \leq \reof (V_{AB})\, ,
\label{protocol_eq16}
\ee
which together with~\eqref{K_tilde_ub} proves~\eqref{KG_ub}.

It remains to prove~\eqref{reof=KG_pure}. Fortunately, this is much easier to do: indeed, it suffices to exhibit a protocol that starting with an arbitrary number of copies of a pure QCM $\gamma_{AB}$ achieves a secret key distillation rate that is arbitrarily close to $M(\gamma_A)$. To do this, fix $\epsilon>0$, and apply 
\eqref{IMc pure states} to select two Gaussian measurements with seeds $\Gamma_A$ and $\Gamma_B$ such that
\bb
I_M(A:B)_{\gamma_{AB}+\Gamma_A\oplus \Gamma_B} \geq M(\gamma_A) - \frac{\epsilon}{3}\, .
\label{protocol_eq17}
\ee
Calling $X$ and $Y$ the outcomes of those measurements, we know that $I(X:Y) = I_M(A:B)_{\gamma_{AB}+\Gamma_A\oplus \Gamma_B}$. Hence, $I(X:Y)\geq M(\gamma_A)-\frac{\epsilon}{3}$ also holds. If Alice and Bob carry out the aforementioned Gaussian measurements separately on every single copy of $\gamma_{AB}$ they share, by applying the above procedure they obtain $n$ independent copies of the jointly Gaussian random variables $X$ and $Y$. Now, let Alice and Bob `bin' the continuous variables $X$ and $Y$ so as to obtain discrete random variables $X'$ and $Y'$ with the property that $I(X':Y') \geq I(X:Y) - \epsilon/3 \geq M(\gamma_A) - 2\epsilon/3$. This is known to be possible~\cite{Kraskov2004}, and indeed can be verified by elementary means, e.g.\ exploiting the uniform continuity of Gaussian distributions.

At this point, we can use a special case of a result proved by Maurer~\cite[Theorem~4]{Maurer1994} (see also previous works by Maurer himself~\cite{Maurer1993} as well as Ahlswede and Csiszar~\cite[Proposition~1]{AhlswedeCsiszar1}), and later generalised to the classical-quantum case in the fundamental work by Devetak and Winter~\cite[Theorem~1]{devetak2005}. For the case where Eve has no prior information, it states that the secret key distillation rate that one can achieve from i.i.d.\ copies of a correlated pair $(X',Y')$ of discrete random variables by means of one-way public communication\footnote{Actually, Maurer's result~\cite[Theorem~4]{Maurer1994} as stated holds for two-way public communication. However, a quick glance at the proof reveals that the two bounds in~\cite[Eq.~(10)]{Maurer1994} require only one-way communication --- either from Alice to Bob or vice versa. Devetak and Winter are more explicit in clarifying that they only need one-way communication~\cite{devetak2005}.} coincides with the mutual information $I(X':Y')$.\footnote{Remember that in our case the variable $Z$, representing Eve's prior information, is absent. Then our claim follows by combining Theorem~4 and the unnumbered equation above~(10) in Maurer's paper~\cite{Maurer1994}.}
To apply Maurer's achievability result, we need to verify that his security criterion is stronger than ours. Writing out everything for the case where Eve has no prior information, a side-by-side comparison of the two security criteria is as follows.
\begin{align}
&\text{Maurer~\cite[Definition~2]{Maurer1994}:} \quad 
\begin{array}{l}\Pr\{S\neq S'\}\leq \epsilon\, , \\
I(S:C)\leq \epsilon\, \\
H(S) \geq \ceil{Rn} - \epsilon\, .
\end{array}
%
\label{Maurer's security criterion} \\
&\text{This paper (Definition~\ref{SK_rate_def}):} \quad 
\inf_{Q_C} \left\| P_{SS'C} - \frac{\delta_{SS'}}{2^{\ceil{Rn}}} \otimes Q_C \right\|_1\leq \epsilon'\, .  \label{our security criterion} 
\end{align}
Here, $|S|$ is the size of the alphabet of $S$, and $\delta_{SS'}/N$ is the perfectly correlated uniform distribution of size $N$. We now verify that~\eqref{Maurer's security criterion} implies~\eqref{our security criterion} for some $\epsilon'$ universally related to $\epsilon$. For the sake of simplicity, we write out the argument in the case where the random variable $C$ ranges over a discrete alphabet. We have that
\begin{align*}
&\inf_{Q_C} \left\| P_{SS'C} - \frac{\delta_{SS'}}{2^{\ceil{Rn}}} \otimes Q_C \right\|_1 \\
&\leq \left\| P_{SS'C} - \frac{\delta_{SS'}}{2^{\ceil{Rn}}} \otimes P_C \right\|_1 \\
&= \sum_{s,s',c} \left| P_{SS'C}(s,s',c) - \frac{\delta_{s,s'}}{2^{\ceil{Rn}}} P_C(c) \right| \\
&= \sum_{s\neq s',\, c} P_{SS'C}(s,s',c) + \sum_{s,c} \left| P_{SS'C}(s,c) - \frac{\delta_{s,s'}}{2^{\ceil{Rn}}} P_C(c) \right| \\
&= \Pr\{S\neq S'\} + \sum_{s,c} \left| P_{SS'C}(s,c) - \frac{1}{2^{\ceil{Rn}}} P_C(c) \right| \\
&\leq \Pr\{S\neq S'\} + \sum_{s,c} \left| P_{SS'C}(s,s,c) - P_{SC}(s,c)\right|\\ 
&\qquad + \sum_{s,c} \left| P_{SC}(s,c) - P_S(s) P_C(c)\right| \\
&\qquad + \sum_{s,c} \left| P_S(s) P_C(c) - \frac{1}{2^{\ceil{Rn}}} P_C(c) \right| \\
&= 2 \Pr\{S\neq S'\} + \left\| P_{SC} - P_S\otimes P_C\right\|_1 + \left\|P_S - \frac{\id}{2^{\ceil{Rn}}}\right\|_1 \\
&\leq 2 \Pr\{S\neq S'\} + \sqrt{2\ln 2 \, I(S:C)} + \sqrt{2\ln 2\left( \ceil{Rn} - H(S)\right)} \\
&\leq 2 \left( \epsilon + \sqrt{2\ln 2\, \epsilon}\right) .
\end{align*}
Here, in the second to last line we applied twice Pinsker's inequality~\cite{Pinsker, Csiszar1967, Kullback1967}, while the last line follows directly from~\eqref{Maurer's security criterion}.


The above argument shows that $I(X':Y')$ is indeed the supremum of all achievable secret key rates for the random variables $(X',Y')$. Therefore, any rate of the form $I(X':Y')-\epsilon/3\geq M(\gamma_A)-\epsilon$ is achievable. Since this holds for an arbitrary $\epsilon>0$, we see that in fact
\bb
M(\gamma_A) \geq \sup \left\{ \text{$R$: $R$ is an achievable rate} \right\} = \kg_\to (\gamma_{AB})\, .
\label{protocol_eq18}
\ee
Together with~\eqref{KG_ub}, this establishes~\eqref{reof=KG_pure} and concludes the proof.\footnote{We note that in the last part of the above proof, we could have equally well leveraged the result of Devetak and Winter~\cite{devetak2005} instead of that by Maurer. 
Verifying that their security criterion is stronger than ours is elementary, and has already been observed e.g.\ by Christandl et al.~\cite{Christandl2007}.
}
\end{proof}

It is important at this point to recall that, when arbitrary local operations are permitted in conjunction with one- or two-way public communication ($1$-LOPC or LOPC, respectively), the secret key of any {\em pure} state $\psi_{AB}$ is well known to equal its local von Neumann entropy $S_1(\psi_A)$, as defined in \eqref{entropy}. Instead,~\eqref{reof=KG_pure} features the R\'enyi-$2$ entropy~\eqref{Renyi-2} of the local state. Since this is typically smaller, $S_2\leq S_1$, our result~\eqref{reof=KG_pure} shows that the Gaussian secret key of any pure Gaussian state is smaller than its unrestricted LOPC secret key, highlighting a fundamental limitation in the ability of Gaussian operations to extract secrecy from quantum states. Later in Section~\ref{sec plob} we will explore an example of such a limitation in a relevant family of {\em mixed} Gaussian states as well.


\subsection{Two-way public communication} \label{Sec:Tm2}

We now turn to our second main result, a weaker bound on the Gaussian secret key of an arbitrary Gaussian state in the presence of two-way public communication.

\begin{thm} \label{KW_ub_thm}
For all QCMs $V_{AB}$, it holds that
\bb
\kg_\leftrightarrow (V_{AB}) \leq 2 \reof(V_{AB})\, .
\label{KW_ub}
\ee
\end{thm}

\begin{proof}
We will prove that $\kg_\leftrightarrow (V_{AB}) \leq \widetilde{K}_\leftrightarrow^\G(V_{AB}) \leq 2 \reof(V_{AB})$, where the first inequality follows from the case $\mathcal{P}=\mathrm{GLOPC}$ of Lemma~\ref{K_tilde_ub_lemma}. To establish the second one, consider as usual a sequence of GLOPC protocols $\Lambda_n: A^n B^n \to \K_n \K'_n \CC_n$ with rate $R$ as in Definition~\ref{modified_SK_rate_def}. Pick numbers $\epsilon_n>0$ such that~\eqref{protocol_eq5} and~\eqref{protocol_eq6} hold, consider an arbitrary sequence of QCMs $\Gamma_{E^n}$, and define the QCM $\tau_{A^nB^n}$ by~\eqref{protocol_eq7}.

Similarly to what we saw in the proof of Theorem~\ref{KG_ub_thm}, since the global input state is Gaussian and all measurements, ancillary states, and unitaries are Gaussian, the whole protocol can be simulated by a purely classical process. The input of this simulation is the pair of correlated random variables $\left(W_{A^n}, W_{B^n}\right)$, whose joint distribution coincides with the Wigner function of $\rho_\G\left[\tau_{A^nB^n},\, 0 \right]$. Let $S_n,S'_n$ be the pair of keys generated by Alice and Bob, and let $C_n$ the messages exchanged. By a result of Maurer, we have that~\cite[Theorem~1]{Maurer1993}
\bb
H(S_n) \leq I\left(W_{A^n} \!: W_{B^n}\right) + H(S_n|S'_n) + I\left(S_n:C_n\right) .
\label{KW_ub_proof_eq1}
\ee
Employing~\eqref{Wigner_entropy_Gaussian} we see immediately that
\bb
\begin{aligned}
I\left(W_{A^n} \!: W_{B^n}\right) &= H(W_{A^n}) + H(W_{B^n}) - H(W_{A^n}W_{B^n}) \\
&= M(\tau_{A^n}) + nm_A \log_2(\pi e) + M(\tau_{B^n}) \\ &\quad + nm_B \log_2(\pi e) - M(\tau_{A^nB^n}) \\ & \quad - n(m_A+m_B)\log_2(\pi e) \\
&= M(\tau_{A^n}) + M(\tau_{B^n}) - M(\tau_{A^nB^n}) \\
&= 2M(\tau_{A^n}) \, ,
\end{aligned}
\label{KW_ub_proof_eq2}
\ee
where the last identity follows because thanks to the discussion following \eqref{elementary_pure_G_lemma} we know that $\tau_{A^nB^n}$ is a pure QCM. Plugging~\eqref{KW_ub_proof_eq2},~\eqref{protocol_eq9},~\eqref{protocol_eq10}, and~\eqref{protocol_eq11} inside~\eqref{KW_ub_proof_eq1} yields
$\ceil{Rn} < 2\epsilon_n \ceil{Rn} + 2\, h_2\left( \frac{\epsilon_n}{2} \right) + h_2(\epsilon_n) + 2\, M\left( \left(\gamma_{AE}^{\oplus n}+\Gamma_{E^n}\right)\Big/\left(\gamma_E^{\oplus n}+\Gamma_{E^n}\right) \right)$,
and in turn 
$\ceil{Rn} < 2\epsilon_n \ceil{Rn} + 2\, h_2\left( \frac{\epsilon_n}{2} \right) + h_2(\epsilon_n) + 2n\, \reof(V_{AB})$
upon taking the infimum over $\Gamma_{E^n}$ as in~\eqref{protocol_eq13}. Dividing by $n$ and taking the limit for $n\to\infty$ gives
$R < 2 \reof(V_{AB})$,
and then
\bb
\widetilde{K}^\G_\leftrightarrow \left( V_{AB} \right) \leq 2 \reof (V_{AB})
\ee
upon an optimisation over all achievable rates $R$.
\end{proof}

We conjecture that the factor of $2$ in~\eqref{KW_ub} is not tight, and that in fact $\kg_\leftrightarrow (V_{AB}) \leq \reof(V_{AB})$ holds true for all QCMs $V_{AB}$. Establishing this would further bolster the operational significance of the Gaussian entanglement measure $\reof$ in the context of QKD. As evidence in favour of our conjecture, we present partial proof of it for the class of protocols GLMPC, corresponding to a scenario where we allow two-way public communication, but only after Alice and Bob perform complete destructive Gaussian measurements on their subsystems. This is the third main result of this paper.

\begin{thm} \label{KGM_ub_thm}
For all QCMs $V_{AB}$, it holds that
\bb
\kgm_\leftrightarrow (V_{AB}) \leq \gier(V_{AB}) \leq \reof(V_{AB})
\label{KGM_ub}
\ee
and moreover
\bb
\gie(V_{AB}) \leq \reof(V_{AB})\, .
\label{gie_ub_reof}
\ee
\end{thm}

The argument we use to prove Theorem~\ref{KGM_ub_thm} is very close in spirit to those proposed by Maurer~\cite{Maurer1993}, Ahlswede and Csiszar~\cite{AhlswedeCsiszar1}, and Maurer and Wolf~\cite{MaurerWolf} to upper bound the secret key capacity of a tripartite probability distribution. In the latter two papers, in particular, the notion of \textit{intrinsic information} was introduced and discussed at length (see~\cite[Theorem~1]{AhlswedeCsiszar1} and~\cite[\S~II]{MaurerWolf}). 

\begin{proof}
We start by proving the first inequality in~\eqref{KGM_ub}. Let Alice, Bob and Eve start with $n$ copies of the pure Gaussian state with QCM $\gamma_{ABE}$, so that the global QCM reads $\gamma_{ABE}^{\oplus n}$. Consider a sequence of GLMPC protocols $\Lambda_n: A^n B^n \to \K_n \K'_n \CC_n$ such that
\begin{align}
&\inf_{\omega_{\CC_n E^n}} \left\| \Lambda_n \left( \rho_\G\left[\gamma_{ABE}^{\otimes n}, 0\right] \right) - \left(\kappa_{2^{\ceil{Rn}}} \right)_{\K_n \K'_n} \otimes \omega_{\CC_n E^n} \right\|_1 < \epsilon_n\, , \label{KGM_ub_proof_eq1}\\
&\lim_{n\to\infty} \epsilon_n = 0\,, \nonumber
\end{align}
for some rate $R>0$, as per Definition~\ref{SK_rate_def}. By construction, the GLMPC protocol $\Lambda_n$ can be decomposed as
\bb
\Lambda_n = \Lambda^c_n\circ \left( \M_{\Gamma_{A^n}}^\G \otimes \M_{\Gamma_{B^n}}^\G \right) ,
\label{KGM_ub_proof_eq2}
\ee
where $\M_{\Gamma_{A^n}}^\G$ and $\M_{\Gamma_{B^n}}^\G$ are complete destructive Gaussian measurements with seeds $\Gamma_{A^n}$ and $\Gamma_{B^n}$ on Alice's and Bob's side, respectively, and $\Lambda^c_n$ is a classical protocol involving only local operations and public communication. For a formal representation of the quantum--classical channels $\M_{\Gamma_{A^n}}^\G, \M_{\Gamma_{B^n}}^\G$, see~\eqref{G_meas_channel}.

Now, consider an arbitrary Gaussian measurement $\M_{\Gamma_{E^n}}^\G$ with seed $\Gamma_{E^n}$ on Eve's subsystem; denote the corresponding output alphabet with $\Z_n = \R^{2nm_E}$. Employing first the data processing inequality for the trace norm and then~\eqref{KGM_ub_proof_eq2} yields
\bb
\begin{aligned}
\epsilon_n &> \inf_{\omega_{\CC_n E^n}} \Big\| \Lambda_n \left( \rho_\G\left[\gamma_{ABE}^{\otimes n}, 0\right] \right) - \left(\kappa_{2^{\ceil{Rn}}} \right)_{\K_n \K'_n} \otimes \omega_{\CC_n E^n} \Big\|_1 \\
&\geq \inf_{Q_{\CC_n \Z_n}} \sup_{\Gamma_{E^n}} \Big\| \left(\Lambda_n \otimes \M_{\Gamma_{E^n}}^\G \right) \left( \rho_\G\left[\gamma_{ABE}^{\otimes n}, 0\right] \right) \\ &\qquad - \left(\kappa_{2^{\ceil{Rn}}} \right)_{\K_n \K'_n} \otimes Q_{\CC_n \Z_n} \Big\|_1 \\
&= \inf_{Q_{\CC_n \Z_n}} \sup_{\Gamma_{E^n}} \Bigg\| \Lambda_n^c\left(\left(\M_{\Gamma_{A^n}}^\G \otimes \M_{\Gamma_{B^n}}^\G \otimes \M_{\Gamma_{E^n}}^\G \right) \left( \rho_\G\left[\gamma_{ABE}^{\otimes n}, 0\right] \right)\right) \\& \qquad - \left(\kappa_{2^{\ceil{Rn}}} \right)_{\K_n \K'_n} \otimes Q_{\CC_n \Z_n} \Big\|_1\, .
\end{aligned}
\label{KGM_ub_proof_eq3}
\ee

Now, the probability distribution $\left(\M_{\Gamma_{A^n}}^\G \otimes \M_{\Gamma_{B^n}}^\G \otimes \M_{\Gamma_{E^n}}^\G \right)$ $\left( \rho_\G\left[\gamma_{ABE}^{\otimes n}, 0\right] \right)$ defines a triple of random variables $X_n,Y_n,Z_n$. Denoting with $C_n$ the message exchanged during the execution of $\Lambda_n^c$ and with $S_n,S'_n$ the locally generated secret keys, the celebrated result of Maurer~\cite[Theorem~1]{Maurer1993} that we have already used multiple times states that
\bb
H(S_n) \leq I(X_n:Y_n|Z_n) + H(S_n|S'_n) + I(S_n:C_n Z_n)\, .
\label{KGM_ub_proof_eq4}
\ee
It is now an elementary exercise to verify that analogous conditions to~\eqref{protocol_eq9}--\eqref{protocol_eq11} apply to our case. The only one which needs a very slight modification is~\eqref{protocol_eq11}. Construct the triple of random variables $\widetilde{S}_n, \widetilde{C}_n, \widetilde{Z}_n$ such that: $\widetilde{S}_n$ and $\widetilde{C}_n, \widetilde{Z}_n$ are independent; $\widetilde{S}_n$ is uniformly distributed over $\left\{1,\ldots, 2^{\ceil{Rn}}\right\}$; and $\widetilde{C}_n, \widetilde{Z}_n$ has probability distribution $Q_{\CC_n \Z_n}$, where
$\Big\| \Lambda_n^c\left(\left(\M_{\Gamma_{A^n}}^\G \otimes \M_{\Gamma_{B^n}}^\G \otimes \M_{\Gamma_{E^n}}^\G \right) \left( \rho_\G\left[\gamma_{ABE}^{\otimes n}, 0\right] \right)\right) - \left(\kappa_{2^{\ceil{Rn}}} \right)_{\K_n \K'_n} \otimes Q_{\CC_n \Z_n} \Big\|_1 \leq \epsilon_n$. 
Then,
\bb
\begin{aligned}
I(S_n:C_n Z_n) &= H(S_n) - H(S_n|C_n Z_n) \\
&\leq H\big(\widetilde{S}_n\big) - H(S_n|C_nZ_n) \\
&= H\big(\widetilde{S}_n \big| \widetilde{C}_n \widetilde{Z}_n \big) - H(S_n|C_nZ_n) \\
&\leq \frac{\epsilon_n}{2} \log_2 \left(2^{\ceil{Rn}}-1\right) + h_2\left(\frac{\epsilon_n}{2}\right) \\
&\leq \frac{\epsilon_n}{2} \ceil{Rn} + h_2\left(\frac{\epsilon_n}{2}\right) .
\end{aligned}
\label{KGM_ub_proof_eq5}
\ee
Also, observe that
\bb
I(X_n:Y_n|Z_n) = I_M \left(A^n : B^n \big| E^n \right)_{\gamma_{ABE}^{\oplus n} + \Gamma_{A^n}\oplus \Gamma_{B^n}\oplus \Gamma_{E^n}}
\label{daje_adapted}
\ee
by~\eqref{daje}. Plugging~\eqref{protocol_eq10},~\eqref{daje_adapted},~\eqref{protocol_eq9}, and~\eqref{KGM_ub_proof_eq5} inside~\eqref{KGM_ub_proof_eq4} yields
\begin{align}
\ceil{Rn} < &\ 2\epsilon_n \ceil{Rn} + 2\, h_2\left( \frac{\epsilon_n}{2} \right) + h_2(\epsilon_n) \nonumber \\ &+ I_M \left(A^n : B^n \big| E^n \right)_{\gamma_{ABE}^{\oplus n} + \Gamma_{A^n}\oplus \Gamma_{B^n}\oplus \Gamma_{E^n}} \, .
\end{align}
Since this holds for all QCMs $\Gamma_{E^n}$,
\begin{align}
\ceil{Rn} < &\ 2\epsilon_n \ceil{Rn} + 2\, h_2\left( \frac{\epsilon_n}{2} \right) + h_2(\epsilon_n) \nonumber \\
&+ \inf_{\Gamma_{E^n}} I_M \big(A^n : B^n \big| E^n \big)_{\gamma_{ABE}^{\oplus n} + \Gamma_{A^n}\oplus \Gamma_{B^n}\oplus \Gamma_{E^n}} \, .
\label{KGM_ub_proof_eq6}
\end{align}
Taking a further supremum on $\Gamma_{A^n}, \Gamma_{B^n}$ and remembering the definition of the Gaussian intrinsic entanglement~\eqref{gie} gives
\bb
\ceil{Rn} < 2\epsilon_n \ceil{Rn} + 2\, h_2\left( \frac{\epsilon_n}{2} \right) + h_2(\epsilon_n) + \gie\left(\gamma_{ABE}^{\oplus n}\right) .
\label{KGM_ub_proof_eq7}
\ee
Dividing by $n$ and taking the liminf for $n\to\infty$ produces
\bb
R < \gier\left(V_{AB}\right) .
\label{KGM_ub_proof_eq8}
\ee
Since this holds for all achievable rates $R$, we also obtain that
\bb
\kgm_\leftrightarrow (V_{AB}) \leq \gier(V_{AB})\, ,
\label{KGM_ub_proof_eq9}
\ee
which proves the first inequality in~\eqref{KGM_ub}.

We now move on to the proof of~\eqref{gie_ub_reof}. To start off, we massage the expression~\eqref{gie} thanks to~\eqref{IM_conditional_Schur}, obtaining
\bb
\gie(V_{AB}) = \sup_{\Gamma_A, \Gamma_B} \inf_{\Gamma_E} I_M\left(A:B\right)_{(\gamma_{ABE} + \Gamma_E)/(\gamma_E + \Gamma_E) + \Gamma_A\oplus \Gamma_B}\, .
\label{gie massaged}
\ee
Before we proceed further, let us define one more quantity via a slight modification of~\eqref{gie massaged}. More precisely, we set
\bb
\gietilde(V_{AB}) \coloneqq \inf_{\text{$\Gamma_E$ pure}}\, \sup_{\Gamma_A,\Gamma_B} I_M(A:B)_{(\gamma_{ABE} + \Gamma_E)/(\gamma_E + \Gamma_E) + \Gamma_A\oplus \Gamma_B}\, .
\label{gie tilde}
\ee
Let us show that $\gie(V_{AB}) \leq \gietilde(V_{AB})$. We have that
\bb
\begin{aligned}
\gie \left( V_{AB}\right) &\leqt{1} \sup_{\Gamma_A, \Gamma_B} \inf_{\text{$\Gamma_E$ pure}} I_M(A:B)_{(\gamma_{ABE} + \Gamma_E)/(\gamma_E + \Gamma_E) + \Gamma_A\oplus \Gamma_B} \\
&\leqt{2} \inf_{\text{$\Gamma_E$ pure}} \sup_{\Gamma_A, \Gamma_B} I_M(A:B)_{(\gamma_{ABE} + \Gamma_E)/(\gamma_E + \Gamma_E) + \Gamma_A\oplus \Gamma_B} \\
&= \gietilde(V_{AB})\, .
\end{aligned}
\label{KGM_ub_proof_eq10}
\ee
Here, in~1 we restricted the infimum in~\eqref{gie massaged} to pure QCMs $\Gamma_{\!E}$, while in~2 we exchanged supremum and infimum according to the max-min inequality $\sup_{x\in \pazocal{X}}\inf_{y\in \pazocal{Y}}f(x,y)\leq \inf_{y\in \pazocal{Y}} \sup_{x\in \pazocal{X}} f(x,y)$~\cite{BOYD}. We now show that in fact $\gietilde(V_{AB}) = \reof (V_{AB})$ holds:
\bb
\begin{aligned}
\gietilde(V_{AB}) &\eqt{3} \inf_{\text{$\Gamma_E$ pure}} I_M^c(A:B)_{(\gamma_{ABE} + \Gamma_E)/(\gamma_E + \Gamma_E)} \\
&\eqt{4} \frac12 \inf_{\text{$\Gamma_E$ pure}}\ I_M (A:B)_{(\gamma_{ABE} + \Gamma_E)/(\gamma_E + \Gamma_E)} \\
&\eqt{5} \frac12 \inf_{\text{$\Gamma_E$ pure}}\ I_M (A:B|E)_{\gamma_{ABE} + \Gamma_E} \\
&\eqt{6} \reof( V_{AB}) .
\end{aligned}
\label{KGM_ub_proof_eq11}
\ee
In~3 we recalled the definition~\eqref{IMc}, in~4 we used~\eqref{IMc pure states}, in~5 we exploited~\eqref{IM_conditional_Schur}, and finally in~6 we leveraged a recently established result on the equality between R\'enyi-$2$ Gaussian squashed entanglement and R\'enyi-$2$ Gaussian entanglement of formation~\cite[Theorem~5 and especially Remark~2]{LL-log-det}.

We have therefore established~\eqref{gie_ub_reof}. Applying it to the QCM $V_{AB}^{\oplus n}$ and taking the liminf for $n\to\infty$ yields
\bb
\begin{aligned}
\gier(V_{AB}) &= \liminf_{n\to\infty} \frac1n \gie \left( V_{AB}^{\oplus n}\right) \\
&\leq \liminf_{n\to\infty} \frac1n \reof \left( V_{AB}^{\oplus n}\right) \\
&= \reof(V_{AB})\, ,
\end{aligned}
\ee
where the last identity is a consequence of the additivity~\eqref{reof additivity} of the R\'enyi-$2$ Gaussian entanglement of formation. This establishes the second inequality in~\eqref{KGM_ub} and concludes the proof.
\end{proof}

Note that the upper bound on $\kgm_\leftrightarrow$ provided by Lemma~\ref{K_tilde_ub_lemma} would lead straight to the inequality $\kgm_\leftrightarrow\leq \reof$. However, this is a priori less tight than the estimate $\kgm_\leftrightarrow \leq \gier$ established in Theorem~\ref{KGM_ub_thm}. This discrepancy is due to the type of constraints we impose on Eve's action: in Lemma~\ref{K_tilde_ub_lemma} we assumed that she performs a destructive Gaussian measurement at the very beginning of the protocol, subsequently broadcasting the obtained outcome to Alice and Bob; in Theorem~\ref{KGM_ub_thm}, instead, we assumed that \emph{first} Alice and Bob make their destructive Gaussian measurements, and \emph{then} Eve makes hers, keeping the outcome secret.


\section{Equivalence of Gaussian entanglement measures}\label{sec:Conj}

In the previous Section we have seen that Theorem~\ref{KGM_ub_thm} brings into play, in addition to the R\'enyi-$2$ Gaussian entanglement of formation~\eqref{reof}, also the (regularised) Gaussian intrinsic entanglement~\eqref{gie}. As mentioned in the Introduction, these two measures have been conjectured to be identical on all Gaussian states~\cite{GIE, GIE-PRA, GIE-PRA2}, and in Theorem~\ref{KGM_ub_thm} we established that at least $\gie\leq \reof$ holds true in general. For a particular --- but, in fact, quite vast --- class of Gaussian states we are able to prove the opposite inequality as well, thus confirming the conjecture.

We call the QCM $V_{AB}$ of a bipartite Gaussian state \emph{normal} if it can be brought into a form in which all $xp$ cross-terms vanish (referred to as $xp$-form) using local symplectic operations alone; see Appendix~\ref{app:XP} for an explicit definition. Here, the $xp$ cross terms of an $m$-mode QCM $W$ are all the entries $W_{jk}$ with $|j-k|>m$, and a local symplectic operation is a map of the form $V_{AB} \mapsto (S_A\oplus S_B) V_{AB} (S_A^\intercal \oplus S_B^\intercal)$, where $S_A, S_B$ are symplectic matrices. All pure QCMs~\cite{giedkemode} as well as all two-mode mixed QCMs~\cite{Simon00, Duan2000} are normal.

Our fourth main result then amounts to the following.
\begin{thm} \label{equality normal thm}
For a normal QCM $V_{AB}$, it holds that
\bb
\gier( V_{AB} ) = \gie( V_{AB} ) = \reof (V_{AB}) .
\label{equality normal}
\ee
In particular,~\eqref{equality normal} holds for all two-mode QCMs.
\end{thm}

The proof of  Theorem~\ref{equality normal thm} makes use of additional facts and technical results presented in Appendices~\ref{sec properties IM} and \ref{sec mammeta}.


\begin{proof}
In the proof of Theorem~\ref{KGM_ub_thm}, and more precisely in~\eqref{gie tilde}, we introduced an auxiliary quantity $\gietilde$. In~\eqref{KGM_ub_proof_eq10} and~\eqref{KGM_ub_proof_eq11} we also showed that $\gie\leq \gietilde = \reof$ on all QCMs.
We now prove that the opposite inequality $\gie (V_{AB})\geq \gietilde (V_{AB})$ holds as well, at least for normal QCMs $V_{AB}$.

Since $V_{AB}$ is normal and both $\gie$ and $\gietilde$ are invariant under local symplectics, we can assume directly that $V_{AB}$ is in  $xp$-form, $V_{AB}=\lsmatrix Q & 0 \\ 0 & P \rsmatrix$. 
Using Lemma~\ref{xp form Williamson lemma} of Appendix~\ref{app:XP}, we construct a symplectic matrix
\bb
S_{AB} = \begin{pmatrix} M^{-1} & 0 \\ 0 & M^\intercal \end{pmatrix} ,
\label{SAB blocks}
\ee
here written with respect to an $xp$ block partition, such that
\bb
S_{AB} V_{AB} S_{AB}^\intercal = \Lambda_{AB} = \begin{pmatrix} D & 0 \\ 0 & D \end{pmatrix} ,
\label{Williamson VAB normal}
\ee
again with respect to the same partition. This is useful because we can now construct very conveniently a purification $\gamma_{ABE}$ of $V_{AB}$:
\bb
\gamma_{ABE} = (S_{AB}\oplus \id_E) \,\gamma_{ABE}^{(0)}\, (S_{AB}\oplus \id_E)^\intercal\, .
\label{gamma ABE 1}
\ee
Here, with respect to an $AB|E$ block partition we have that
\bb
\gamma_{ABE}^{(0)} = \begin{pmatrix} \Lambda & \sqrt{\Lambda^2-\id}\, \Sigma \\ \sqrt{\Lambda^2-\id}\, \Sigma & \Lambda \end{pmatrix} .
\label{gamma ABE 2}
\ee
where $\Sigma\coloneqq (\Pi^x)^\intercal \Pi^x - (\Pi^p)^\intercal \Pi^p $. An important observation to make is that $\gamma_{ABE}^{(0)}$, $S_{AB}$, and hence also $\gamma_{ABE}$ are all in $xp$-form.

Let us now go back to the sought inequality $\gie (V_{AB})\geq \gietilde (V_{AB})$. Since the left-hand side is defined by a supremum over Gaussian measurements, parametrised by $\Gamma_A, \Gamma_B$, we estimate it as
\bb
\begin{aligned}
\gie (V_{AB}) &\eqt{1} \sup_{\Gamma_A, \Gamma_B} \inf_{\Gamma_E} I_M \left(A:B\right)_{(\gamma_{ABE} + \Gamma_E)/(\gamma_E + \Gamma_E) + \Gamma_A\oplus \Gamma_B} \\
&\geqt{2} \lim_{t\to\infty} \inf_{\Gamma_E} I_M\left(A:B\right)_{(\gamma_{ABE} + \Gamma_E)/(\gamma_E + \Gamma_E) + \Gamma_A(t) \oplus \Gamma_B(t)} \\
&\eqt{3} \inf_{\Gamma_E} I_M\left(A_x:B_x\right)_{((\gamma_{ABE} + \Gamma_E)/(\gamma_{E} + \Gamma_E))^x}\, .
\end{aligned}
\label{intermediate_lb_gie}
\ee
In the above derivation, the equality in~1 is simply~\eqref{gie massaged}, the inequality in~2 follows by choosing $\Gamma_A(t), \Gamma_B(t)$ as defined in~\eqref{seeds homodyne}, and the inequality in~3, which is the real technical hurdle here, follows by combining Lemma~\ref{measured QCMs compact} and Proposition~\ref{homodyne limit prop} of Appendix~\ref{sec mammeta}.

We now look at the set of matrices $((\gamma_{ABE} + \Gamma_E)/(\gamma_E + \Gamma_E))^x$, where $\Gamma_E$ is an \emph{arbitrary} QCM on $E$, not necessarily in $xp$-form. With respect to an $xp$ block partition, let us parametrise it as
\bb
\Gamma_E = \begin{pmatrix} K & J \\ J^\intercal & L \end{pmatrix} .
\label{Gamma E blocks}
\ee
We now compute:
\begin{align*}
&\big((\gamma_{ABE} + \Gamma_E)/(\gamma_E + \Gamma_E)\big)^x \\ &= \Pi^x \big((\gamma_{ABE} + \Gamma_E)/(\gamma_E + \Gamma_E)\big) (\Pi^x)^\intercal \\[0.5ex]
&\eqt{4} \Pi^x (S_{AB}\oplus \id_E) \left(\big(\gamma_{ABE}^{(0)} + \Gamma_E\big)\big/\big(\gamma_E^{(0)} + \Gamma_E\big)\right) (S_{AB}\oplus \id_E)^\intercal (\Pi^x)^\intercal \\[0.5ex]
&\eqt{5} M^{-1} \Pi^x \left(\big(\gamma_{ABE}^{(0)} + \Gamma_E\big)\big/\big(\gamma_E^{(0)} + \Gamma_E\big)\right) (\Pi^x)^\intercal M^{-\intercal} \\[0.5ex]
&\eqt{6} M^{-1} \Pi^x \left( \Lambda -\!\! \sqrt{\Lambda^2-\id}\, \Sigma\, (\Lambda+\Gamma_E)^{-1} \Sigma\, \!\! \sqrt{\Lambda^2-\id} \right) (\Pi^x)^\intercal M^{-\intercal} \\[0.5ex]
&\eqt{7} M^{-1} \left( D -\!\! \sqrt{D^2-\id}\, \Pi^x\, (\Lambda+\Gamma_E)^{-1} (\Pi^x)^\intercal\,\!\! \sqrt{D^2-\id} \right) M^{-\intercal} \\[0.5ex]
&\eqt{8} M^{-1} \left( D -\!\! \sqrt{D^2-\id} \left( K\! +\! D\! -\! J (L\! + \! D)^{-1}\! J^\intercal \right)^{-1}\!\! \sqrt{D^2-\id} \right) M^{-\intercal}.
\end{align*}
The justification of the above derivation is as follows. 4: We made use of~\eqref{gamma ABE 1} and of the covariance property~\eqref{Schur covariance}. 5: We applied~\eqref{SAB blocks}. 6: We computed the Schur complement with the help of~\eqref{gamma ABE 2}. 7: We used the fact that $\Lambda$ and $\Sigma$ are all in $xp$-form. 8: We calculated
\bbb
(\Lambda+\Gamma_{\!E})^{-1} = \begin{pmatrix} D + K & J \\ J^\intercal & D + L \end{pmatrix}  = \begin{pmatrix} \left( K + D - J (L+D)^{-1} J^\intercal \right)^{-1} & \!\!* \\ \!\!* & \!\!* \end{pmatrix}
\eee
thanks to the block inversion formulae~\eqref{inv}. Here, the symbols $*$ indicate unspecified matrices of appropriate size.

By the above calculation, the function on the right-hand side of~\eqref{intermediate_lb_gie} depends only on the combination $K-J(L+D)^{-1}J^\intercal$, which is a rather special function of the free variable $\Gamma_E$. Now, on the one hand
\bbb
K-J(L+D)^{-1}J^\intercal = (\Gamma_E + 0\oplus D)/(L+D) \geq \Gamma_E/L >0\, ,
\eee
where we applied the monotonicity of Schur complements~\eqref{Schur monotonicity}, and then observed that $\Gamma_E/L>0$ follows from the block positivity conditions~\eqref{Schur block positivity}, once one remembers that $\Gamma_E>0$ as $\Gamma_E$ is a QCM. On the other hand, every positive definite matrix $T>0$ can be written as $T=K-J(L+D)^{-1}J^\intercal$ for some $K,J,L$ forming --- according to~\eqref{Gamma E blocks} --- a \emph{pure} QCM $\Gamma_E$ in $xp$-form. In fact, it suffices to set $K=T=L^{-1}$ and $J=0$; this makes the corresponding $\Gamma_E$ pure, as can be seen e.g.\ by comparing it with~\eqref{pure in xp form}, and clearly in $xp$-form. We have just proved that
\begin{align}
&\left\{ K-J(L+D)^{-1}J^\intercal:\, \Gamma_E\geq i\Omega_E \right\} = \left\{T:\, T>0\right\}  \\
&= \left\{ K-J(L+D)^{-1}J^\intercal:\, \text{$\Gamma_E$ pure QCM in $xp$-form} \right\} ,\nonumber
\end{align}
which can be rephrased as
\begin{align}
&\left\{ ((\gamma_{ABE} + \Gamma_E)/(\gamma_E + \Gamma_E))^x:\, \Gamma_E\geq i\Omega_E \right\} \label{crucial identity} \\
&= \left\{ ((\gamma_{ABE} + \Gamma_E)/(\gamma_{E} + \Gamma_E))^x:\, \text{$\Gamma_E$ pure QCM in $xp$-form} \right\} . \nonumber
\end{align}
We make use of this crucial fact to further massage the right-hand side of~\eqref{intermediate_lb_gie}, obtaining that
\bb
\begin{aligned}
\gie (V_{AB}) &\geq\, \inf_{\Gamma_E} I_M\left(A_x:B_x\right)_{((\gamma_{ABE} + \Gamma_E)/(\gamma_E + \Gamma_E))^x} \\
&\eqt{9} \inf_{\substack{\\ \text{$\Gamma_E$ pure} \\ \text{in $xp$-form}}} I_M(A_x : B_x)_{((\gamma_{ABE} + \Gamma_E)/(\gamma_E + \Gamma_E))^x} \\
&\eqt{10} \frac12 \inf_{\substack{\\ \text{$\Gamma_E$ pure} \\ \text{in $xp$-form}}} I_M(A : B)_{(\gamma_{ABE} + \Gamma_E)/(\gamma_E + \Gamma_E)} \\
&\eqt{11} \inf_{\substack{\\ \text{$\Gamma_E$ pure} \\ \text{in $xp$-form}}} \sup_{\Gamma_A, \Gamma_B} I_M(A: B)_{(\gamma_{ABE} + \Gamma_E)/(\gamma_E + \Gamma_E) + \Gamma_A\oplus \Gamma_B} \\
&\geqt{12} \inf_{\text{$\Gamma_E$ pure}} \sup_{\Gamma_A, \Gamma_B} I_M(A : B)_{(\gamma_{ABE} + \Gamma_E)/(\gamma_E + \Gamma_E) + \Gamma_A\oplus \Gamma_B} \\
&= \gietilde(V_{AB})\, .
\end{aligned}
\label{basically_done}
\ee
Here: in~9 we used~\eqref{crucial identity}; in~10 we observed that both $\gamma_{ABE},\Gamma_E$, and hence also the pure QCM $(\gamma_{ABE} + \Gamma_E)/(\gamma_E + \Gamma_E)$, are in $xp$-form, which allowed us to apply Corollary~\ref{pure_xp_IM_cor} of Appendix~\ref{sec properties IM}; in~11 we recalled~\eqref{IMc pure states}, and finally~12 follows elementarily from enlarging the set over which we compute the infimum. Combining the above inequality with~\eqref{KGM_ub_proof_eq10} and~\eqref{KGM_ub_proof_eq11} proves that $\gie(V_{AB})=\reof(V_{AB})$ for all normal QCMs $V_{AB}$.

To complete the proof, it suffices to observe that the direct sum $V_{AB}^{\oplus n}$ of normal matrices $V_{AB}$ is still normal. Hence,
\begin{align*}
\gie(V_{AB}) &= \liminf_{n\to\infty} \frac1n \gie\left( V_{AB}^{\oplus n} \right) \\
&= \liminf_{n\to\infty} \frac1n \reof\left( V_{AB}^{\oplus n} \right) = \reof(V_{AB})\, ,
\end{align*}
where the last equality comes from~\eqref{reof additivity}. 
\end{proof}

Theorem~\ref{equality normal thm} establishes a powerful equivalence of two originally quite distinct Gaussian entanglement measures, for all two-mode Gaussian states and more generally all normal QCMs. The reader could wonder whether normal QCMs constitute a proper subset of all QCMs beyond the two-mode case, In  Appendix~\ref{sec non-normal} we show that this is indeed the case, by constructing an explicit example of a non-normal QCM over a $(1+2)$-mode system. The validity of the conjecture $\gier(V_{AB}) \eqt{?} \reof(V_{AB})$ for non-normal QCMs $V_{AB}$ remains  open in general.

\section{Applications and examples}\label{sec:EX}

\subsection{Secret key from noisy two-mode squeezed states}\label{sec plob}

We now apply our results, and in particular Theorem~\ref{KG_ub_thm}, to study secret key distillation from a class of Gaussian states of immediate physical interest.  The states we will look at are obtained by sending one half of a two-mode squeezed vacuum $\ket{\psi_s}$ across a pure loss channel $\pazocal{E}_\lambda$ (a.k.a.\ a quantum-limited attenuator). We recall that a two-mode squeezed vacuum is defined by
\bb
\begin{aligned}
\ket{\psi_s} &\coloneqq \frac{1}{\cosh(r_s)} \sum_{n=0}^\infty \tanh^n r_s \ket{nn}\, 
\end{aligned}
\label{tmsv}
\ee
where $\ket{n}$ denote local Fock states, and the squeeze parameter $r_s\coloneqq \frac{s \ln 10}{20}$ is expressed as a function of the squeezing intensity $s$ measured in $\si{\decibel}$.
The  pure loss channel $\pazocal{E}_\lambda$ is a Gaussian channel whose action at the level of density operators can be expressed as
\begin{align}
\pazocal{E}_\lambda (\rho) &\coloneqq \Tr_2 \left[ \pazocal{U}_\lambda\left( \rho\otimes \ketbra{0}\right) \pazocal{U}_\lambda^\dag \right] , \label{rho_lambda_s_pure_loss}
\end{align}
where $\pazocal{U}_\lambda \coloneqq e^{i\arccos\sqrt\lambda\, \left(x_1 p_2 - x_2 p_1\right)}$ is the Gaussian unitary that represents the action of a beam splitter with transmissivity $\lambda$, $\Tr_2$ stands for the partial trace over the second mode, and as usual $x_j,p_j$ denote the canonical operators pertaining to the $j^{\text{th}}$ mode.

The R\'enyi-$2$ Gaussian entanglement of formation of the state $\left(\pazocal{E}_\lambda \otimes I\right)(\psi_s)$ can be expressed in closed form by adapting the results for the standard (von Neumann) Gaussian entanglement of formation, which has been computed in \cite{Tserkis18}. We find
\bb
\reof\left( \left(\pazocal{E}_\lambda \otimes I\right)(\psi_s)\right) = \log_2 \left(\frac{1+(1+\lambda)\sinh^2 r_s}{1+(1-\lambda)\sinh^2 r_s} \right) .
\label{reof_ex}
\ee

No expression for the corresponding $1$-LOPC secret key $K_\to\left(\left(\pazocal{E}_\lambda \otimes I\right)(\psi_s)\right)$ seems to be known. However, in order to demonstrate the effectiveness of our estimate~\eqref{KG_ub}, it suffices to consider suitable lower bounds on this quantity. One such bound is the one-way distillable entanglement~\cite{Bennett-distillation, Bennett-error-correction, devetak2005}, denoted with $D_\to$. We succeeded in computing $D_\to\left(\left(\pazocal{E}_\lambda \otimes I\right)(\psi_s)\right)$ because the state in question is `degradable', and hence its one-way distillable entanglement equals the readily found coherent information~\cite{Leditzky-Nila}. The resulting expression is
\bb
\begin{aligned}
K_\to\!\left((\pazocal{E}_\lambda\! \otimes\! I)(\psi_s)\right) &\geq D_\to\!\left((\pazocal{E}_\lambda\! \otimes\! I)(\psi_s)\right) \\
&= g\left(\sinh^2 r_s\right) - g\left((1\!-\!\lambda)\sinh^2 r_s \right) ,
\end{aligned}
\label{Ed_ex}
\ee
where $g(x)\coloneqq (x+1)\log_2(x+1) - x \log_2(x)$ is the \emph{bosonic entropy function}.

\begin{figure}[t]
\centering
\includegraphics[scale=.8]{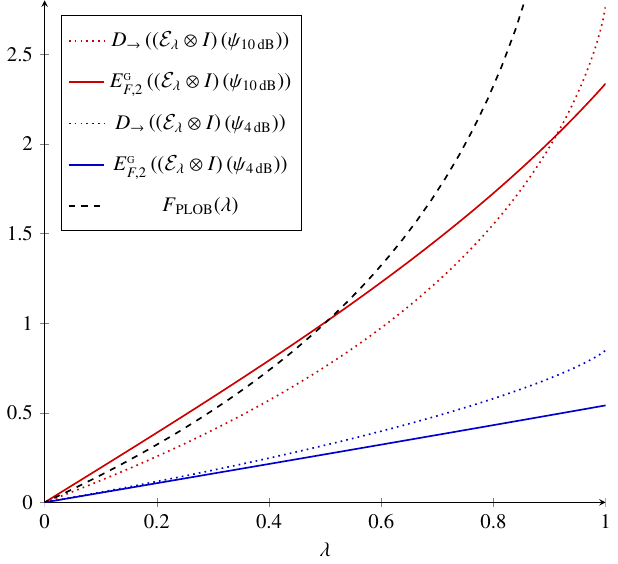}
\caption{
A comparison between the one-way distillable entanglement (communication direction: first-to-second subsystem), which lower bounds the 1-LOPC secret key rate $K_\to$, and the \mbox{R\'enyi-$2$} Gaussian entanglement of formation, which upper bounds the \emph{Gaussian} 1-LOPC secret key rate. Both functions are computed for the states $\left(\pazocal{E}_\lambda \otimes I\right)(\psi_s)$. Explicit formulae are reported in~\eqref{Ed_ex} and~\eqref{reof_ex}, respectively. The comparison proves that the restriction to Gaussian operations reduces the distillable secret key appreciably for either low squeezing or high transmissivity. Our upper bound on the secret key rate can be and often is tighter than the one presented in~\cite{Pirandola17}, equal to $F_{\mathrm{PLOB}}(\lambda)\coloneqq -\log_2(1-\lambda)$ and plotted here as the black dashed line. This is because the bound in~\cite{Pirandola17} applies to all (Gaussian and non-Gaussian) protocols with unbounded energy budget, while ours is tailored to Gaussian protocols and takes into account limitations to the available energy.
}
\label{reof_vs_ed_fig}
\end{figure}

In Figure~\ref{reof_vs_ed_fig} we compare the upper bound~\eqref{reof_ex} for $\kg_\to \left( \left(\pazocal{E}_\lambda \otimes I\right)(\psi_s)\right)$ deduced from Theorem~\ref{KG_ub_thm} and the lower bound~\eqref{Ed_ex} for $K_\to \left( \left(\pazocal{E}_\lambda \otimes I\right)(\psi_s)\right)$. The plots show that the former quantity is smaller than the latter for all $\lambda\in (0,1]$ when $s \leq s_0\approx \SI{4.22}{\decibel}$, and only for sufficiently large $\lambda\geq \lambda_0(s)$ when $s>s_0$. For example, for the nowadays experimentally feasible~\cite{Vahlbruch-10dB, Vahlbruch-15dB} value of $s=\SI{10}{\decibel}$, the useful range becomes $\lambda \geq \lambda_0(\SI{10}{\decibel}) \approx 0.912$. 

This shows that, in several physically interesting regimes (either low squeezing or high transmissivity), our bounds accurately capture and quantify the severity of the Gaussian restriction for the task of distilling secrecy.

\subsection{A conditional mutual information game}\label{sec game}

Finally, we interpret our results in a game-theoretical context. 
We begin by observing, as an interesting side result, that our proof of Theorem~\ref{equality normal thm} implies the following variant of the strong saddle-point property of the log-determinant conditional mutual information.

\begin{prop}[(Saddle-point property of log-determinant conditional mutual information)]
Let $V_{AB}$ be a normal QCM with purification $\gamma_{ABE}$. Then
\bb
\begin{aligned}
\reof(V_{AB}) &= \inf_{\Gamma_{\!E}}\sup_{\Gamma_{\!A}, \Gamma_{\!B}}I_M\left(A\!:\!B|E\right)_{\gamma_{ABE} + \Gamma_{\!A}\oplus \Gamma_{\!B}\oplus \Gamma_{\!E}} \\
&= \sup_{\Gamma_{\!A}, \Gamma_{\!B}} \inf_{\Gamma_{\!E}} I_M\left(A\!:\!B|E\right)_{\gamma_{ABE} + \Gamma_{\!A}\oplus \Gamma_{\!B}\oplus \Gamma_{\!E}}\, .
\end{aligned}
\label{minimax_equality}
\ee
\end{prop}

\begin{proof}
The top-most expression is obviously greater or equal to the bottom-most one, owing to the max-min inequality $\sup_{x\in \pazocal{X}}\inf_{y\in \pazocal{Y}}f(x,y)\leq \inf_{y\in \pazocal{Y}} \sup_{x\in \pazocal{X}} f(x,y)$.
For the opposite inequality, let us write
\begin{align*}
    \sup_{\Gamma_A, \Gamma_B} & \inf_{\Gamma_E} I_M\left(A:B|E\right)_{\gamma_{ABE} + \Gamma_A\oplus \Gamma_B\oplus \Gamma_E} = \gie(V_{AB}) \\
    &\geqt{1} \gietilde(V_{AB}) \\
    &= \inf_{\text{$\Gamma_E$ pure}}\, \sup_{\Gamma_A,\Gamma_B} I_M(A:B)_{(\gamma_{ABE} + \Gamma_E)/(\gamma_E + \Gamma_E) + \Gamma_A\oplus \Gamma_B} \\
    &\geqt{2} \inf_{\Gamma_E}\, \sup_{\Gamma_A,\Gamma_B} I_M(A:B)_{(\gamma_{ABE} + \Gamma_E)/(\gamma_E + \Gamma_E) + \Gamma_A\oplus \Gamma_B}
\end{align*}
Here, 1~follows from~\eqref{basically_done}, while 2~can be deduced by noting that infimum has been enlarged. The proof is then complete.
\end{proof}

Equalities of the form~\eqref{minimax_equality} represent a quintessence of application of methods of game theory in information theory. They appear in the context of a generic problem of finding optimal strategies for communication over a jamming channel. The task is linked to game theory by interpreting the communication as a two-player game between the sender-receiver pair on the one hand, and the malicious jammer on the other; here, the payoff function is some information measure, typically the mutual information~\cite{Borden1985, Stark1988, CT}. The goal of the sender-receiver pair is to maximise the payoff function, whereas the goal of the jammer is to minimise it. If the payoff function exhibits a saddle-point property akin to~\eqref{minimax_equality} on the sets of allowed strategies of the players, then the saddle-point strategies are simultaneously optimal for both players. The game is then said to have a {\it value}, which is equal to the saddle-point value of the payoff function.

Viewed from a game-theoretical perspective, Equation~\eqref{minimax_equality} then ensures the existence of a value of the following Gaussian quantum game with log-determinant conditional mutual information as the payoff function. At the beginning of the game, the players share a fixed pure Gaussian state with QCM
\bb
\label{gammaABE}
\gamma_{ABE} = \begin{pmatrix}
V_{AB} & V_{ABE}\\
V_{ABE}^{T} & V_{E}\\
\end{pmatrix} .
\ee
Clearly, we can see $\gamma_{ABE}$ as a purification of the state with QCM $V_{AB}$. The participants holding subsystems $A$ and $B$, called Alice and Bob in what follows, choose Gaussian measurements characterised by QCMs $\Gamma_{A}$ and $\Gamma_{B}$ to maximise the conditional mutual information $I_M\left(A\!:\!B|E\right)_{\gamma_{ABE} + \Gamma_{\!A}\oplus \Gamma_{\!B}\oplus \Gamma_{\!E}}$, while the jammer Eve holding subsystem $E$ chooses a Gaussian measurement with QCM $\Gamma_{E}$ to minimise it. The equality~\eqref{minimax_equality} then guarantees that such a game has a value, and that this value is equal to $\reof (V_{AB})$ by~\eqref{gie} and~\eqref{equality normal}. 

However, the game does not have the structure of a typical communication game with jamming. Namely, all participants appear symmetrically in the game and, in particular, it is not clearly seen, how the jammer disturbs the communication channel between the sender and the receiver. Nevertheless, we can transform the game into a teleportation game (different from the one presented in~\cite{Pirandola2005}) exhibiting all the features mentioned above. As a bonus, the obtained game reveals how the separability properties of the initial state across the $A:B$ partition and the measurement chosen by the jammer influence the effective state shared by Alice and Bob.

\begin{figure}[t]
\begin{center}
\includegraphics[width=\columnwidth]{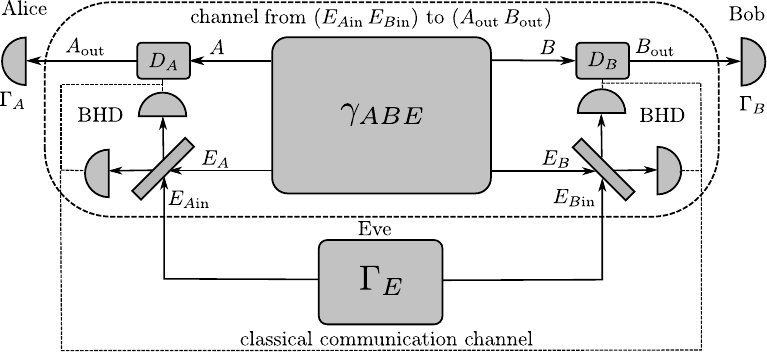}
\caption{Scheme of the teleportation game. The jammer Eve sends through a channel a bipartite Gaussian state of subsystem $E=E_{A{\rm in}}E_{B{\rm in}}$ with zero first moments and QCM $\Gamma_{E}$ to subsystems $A_{\rm out}$ and $B_{\rm out}$ of spatially separated recipients Alice and Bob. The channel transforms QCM $\Gamma_{E}$ as in Eq.~\eqref{Schur_sigma} and it is realised by teleportation (dashed rectangular block) from Eve to Alice and Bob with the help of a shared fixed Gaussian state with QCM $\gamma_{ABE}$, Eq.~\eqref{gammaABE}. Here, the shortcut BHD stands for the balanced homodyne detection and $D_{A}, D_{B}$ are displacements. Alice and Bob finally perform Gaussian measurements on their subsystems $A_{\rm out}$ and $B_{\rm out}$ characterised by QCMs $\Gamma_{A}$ and $\Gamma_{B}$. The payoff function is the log-determinant conditional mutual information $I_M\left(A\!:\!B|E\right)_{\gamma_{ABE} + \Gamma_{\!A}\oplus \Gamma_{\!B}\oplus \Gamma_{\!E}}$ and Alice and Bob choose their measurement QCMs so as to maximise the payoff function, while Eve aims at minimising it. Interestingly, such a game has a value and it is given by the R\'enyi-$2$ Gaussian entanglement of formation $\reof (V_{AB})$ of the reduced QCM $V_{AB}$ of QCM $\gamma_{ABE}$, which belongs to subsystems $A$ and $B$.}
\label{fig_game}
\end{center}
\end{figure}

To find the latter game, we first rewrite equality~\eqref{minimax_equality} as
\begin{equation*}
\inf_{\Gamma_{\!E}}\sup_{\Gamma_{\!A}, \Gamma_{\!B}}I_M\left(A\!:\!B\right)_{\sigma_{AB} + \Gamma_{\!A}\oplus \Gamma_{\!B}}
= \sup_{\Gamma_{\!A}, \Gamma_{\!B}} \inf_{\Gamma_{\!E}} I_M\left(A\!:\!B\right)_{\sigma_{AB} + \Gamma_{\!A}\oplus \Gamma_{\!B}}\, ,
\end{equation*}
where 
\begin{equation}
\sigma_{AB}=V_{AB}-V_{ABE}\left({V_{E}+R\Gamma_{E} R}\right)^{-1} V_{ABE}^\intercal\, ,
\label{Schur_sigma}
\end{equation}
with $R=\mathrm{diag}(1,-1,1,-1,\ldots 1,-1)$ being the diagonal matrix representing on the QCM level the transposition operation $x_{j}\rightarrow x_j^\intercal = x_j$, $p_{j}\rightarrow p_j^\intercal = -p_j$. Here, we used property~\eqref{IM_conditional_Schur} of the conditional mutual information, together with the fact that $R\Gamma_{\!E}R$ is a physical QCM, which runs over the set of all QCMs as $\Gamma_{\!E}$ is varied over the set of all QCMs.

Looking closely at the Schur complement $\sigma_{AB}$, Eq.~\eqref{Schur_sigma}, one further finds~\cite{nogo2,nogo3} that it can be viewed as an output of a Gaussian trace-decreasing completely-positive map characterised by the QCM~\eqref{gammaABE} with the state with QCM $\Gamma_{\!E}$ at the input. Here, $E$ labels the input system and $AB$ the output system. Since the map can be implemented deterministically~\cite{nogo2} via the standard CV teleportation protocol~\cite{Braunstein1998}, we arrive at the teleportation scheme in Figure~\ref{fig_game}.

In view of the saddle-point property~\eqref{minimax_equality}, the depicted teleportation game has a value, which is given exactly by the R\'enyi-$2$ Gaussian entanglement of formation $\reof (V_{AB})$. The optimal strategy for both players is to choose QCMs $\Gamma_{A}, \Gamma_{B}$ and $\Gamma_{E}$, which achieve $\reof (V_{AB})$.   This result makes $\reof(V_{AB})$ a unique instance of an entanglement measure equipped with such a game-theoretical interpretation.

\section{Conclusions} \label{sec:End}
We studied the operational task of distilling a secret key from Gaussian states using local Gaussian operations, local classical processing, and public communication. When only one-way public communication is allowed, we determined the exact expression of the Gaussian secret key for all Gaussian pure states, and established upper bounds that hold for mixed multi-mode Gaussian states in all other cases. These bounds can be used to benchmark state-of-the-art CV QKD protocols against much simpler, Gaussian ones. 
Our findings imply that Gaussian secret key distillation, albeit often possible with positive yield, can be strictly less efficient than a general protocol would be. 
In the Gaussian-restricted scenario, our results often tighten the bounds obtained using the squashed entanglement~\cite{Takeoka14} and the relative entropy of entanglement~\cite{Pirandola17,Pirandola18}.

We also proved a recently proposed conjecture~\cite{GIE} on the equality between Gaussian intrinsic entanglement and R\'enyi-$2$ Gaussian entanglement of formation for all Gaussian states whose covariance matrix is `normal', and in particular for all two-mode Gaussian states. In conjunction with the already proven equality between the latter measure and a Gaussian version of the squashed entanglement~\cite{LL-log-det}, this establishes a coalescent and strongly operationally motivated highway to quantifying entanglement of Gaussian states. We further presented an alternative operational interpretation for this {\it treble} entanglement quantifier in a game-theoretical scenario. The unification presented in this paper stands in stark contrast with the recently uncovered fundamental non-uniqueness of general entanglement measures~\cite{irreversibility, gap}, and points to a much simpler picture in the special case of Gaussian entanglement.

\begin{acknowledgments}
LL and LM contributed equally to this work. LL was supported by the Alexander von Humboldt Foundation. GA acknowledges support by the European Research Council (ERC) under the Starting Grant GQCOP (Grant no.~637352) and by the UK Research and Innovation (UKRI) under BBSRC Grant BB/X004317/1 and EPSRC Grant EP/X010929/1. GA thanks S.~Tserkis and M.~Gideon for fruitful discussions.
\end{acknowledgments}

\bibliography{biblionew}

\begin{thebibliography}{132}
\bibitem [{not({\natexlab{a}})}]{notemean}%
  \BibitemOpen
  \href@noop {} {} ({\natexlab{a}}),\ \bibinfo {note} {we usually assume $s=0$,
  since the mean can be changed by local unitaries (displacements), and these
  do not affect the entanglement nor the (Gaussian) secret key properties of
  the state.}\BibitemShut {Stop}%
\bibitem [{not({\natexlab{b}})}]{notedirect}%
  \BibitemOpen
  \href@noop {} {} ({\natexlab{b}}),\ \bibinfo {note} {the direct sum is
  understood to correspond to a partition of the canonical operators into those
  pertaining to each subsystem. For instance, for a two-mode bipartite system
  we have the decomposition $r_{AB} = (x_A, x_B, p_A, p_B)^\intercal = (x_A,
  p_A)^\intercal \oplus (x_B, p_B)^\intercal= r_A \oplus r_B$.}\BibitemShut
  {Stop}%
\end{thebibliography}


\allowdisplaybreaks
\appendix

\makeatletter

\setcounter{secnumdepth}{2}

\section{Additional notation and definitions}\label{app:Gauss}

\subsection{Subgroups of the symplectic group}

We start by fixing some notation. Recall that the symplectic form of an $n$-mode system takes the form
\bb
\Omega = \begin{pmatrix} 0 & \id \\ -\id & 0 \end{pmatrix},
\label{Omega}
\ee
where all submatrices are $n\times n$. Let us also define
\bb
H\coloneqq \frac{1}{\sqrt2} \begin{pmatrix} \id & \id \\ i \id & -i \id \end{pmatrix}
\label{H}
\ee
as the unitary matrix that diagonalises $\Omega$, i.e.\ such that
\bb
H^\dag \Omega H = \begin{pmatrix} i\id & 0 \\ 0 & -i \id \end{pmatrix} .
\ee
The symplectic group $\mathbb{SP}(n)$ is formed by all those $2n\times 2n$ real matrices that preserve the symplectic form $\Omega$:
\bb
\mathbb{SP}(n) \coloneqq \left\{ S:\, S\Omega S^\intercal = \Omega \right\} .
\label{Sp}
\ee
The $2n\times 2n$ symplectic orthogonal matrices form a subgroup $\mathbb{K}(n)\subset \mathbb{SP}(n)$ that is isomorphic to the unitary group $\mathbb{U}(n)$. The isomorphism is given by~\cite[\S~2.1.2]{GOSSON}
\bb
\mathbb{U}(n) \ni U \longleftrightarrow K = H \begin{pmatrix} U & 0 \\ 0 & U^* \end{pmatrix} H^\dag \in \mathbb{K}(n)\, .
\label{KU isomorphism}
\ee
Another important subgroup of the symplectic group is isomorphic to $\mathbb{GL}(n)$, the group of $n\times n$ invertible real matrices~\cite[Example~2.5]{GOSSON}:
\bb
\mathbb{GL}(n) \simeq \left\{ \begin{pmatrix} M^{-1} & 0 \\ 0 & M^\intercal \end{pmatrix}:\ \text{$M$ invertible} \right\} \subset \mathbb{SP}(n) \, .
\label{GL isomorphism}
\ee

\subsection{Normal covariance matrices}\label{app:XP}

The phase space of an $m$-mode system is usually divided into its first $m$ and last $m$ components, called the $x$ and $p$ components, respectively. Let us denote by $\Pi^x:\R^{2m}\to \R^{m}$ and $\Pi^p:\R^{2m}\to \R^{m}$ the corresponding orthogonal projectors. Accordingly, we can write e.g.
$\Pi^x \Omega (\Pi^x)^\intercal = 0_n = \Pi^p \Omega (\Pi^p)^\intercal$, $\Pi^x \Omega (\Pi^p)^\intercal = \id_n = - \Pi^p \Omega (\Pi^x)^\intercal$.
Observe that we think of $\Pi^x$ as the $m\times 2m$ rectangular matrix $\lsmatrix \id & 0 \rsmatrix$, and correspondingly of $(\Pi^x)^\intercal$ as the $2m\times m$ rectangular matrix $\lsmatrix \id \\ 0 \rsmatrix$. A matrix $V$ expressed as
\bb
V= \begin{pmatrix} \Pi^x V(\Pi^x)^\intercal & \Pi^x V(\Pi^p)^\intercal \\ \Pi^p V(\Pi^x)^\intercal & \Pi^p V(\Pi^p)^\intercal \end{pmatrix}
\ee
is said to be written with respect to the $xp$ block partition. The representations in~\eqref{Omega} and~\eqref{H} are of this form.

This somehow pleonastic notation comes in handy when multiple systems are involved. In that case, we adopt the convention of indicating as subscripts the systems ($A$, $B$, and so on) and as superscripts the phase space components we want to project onto ($x$ and $p$). For example, the projector onto the $x$ component of the $B$ system will be denoted by $\Pi^x_B$.

Particularly simple matrices are those that have no $xp$ cross-terms. We say that a matrix $V$ is in \textit{${xp}$-form} if
\bb
\Pi^x V (\Pi^p)^\intercal = 0\, .
\label{xp form}
\ee
Observe that since $V$ is symmetric, from the above identity it also follows that $\Pi^p V (\Pi^x)^\intercal = 0$. With respect to an $xp$ block partition, a matrix in $xp$-form then reads
\bb
V = \begin{pmatrix} Q & 0 \\ 0 & P \end{pmatrix} .
\label{V xp}
\ee
When $V$ is a QCM, it is not difficult to verify that the bona fide condition $V\geq i \Omega$ implies that $Q,P>0$ and $Q\geq P^{-1}$. Williamson's theorem~\cite{willy} states that all QCMs can be brought into a special kind of $xp$-form via symplectic congruence: a QCM is in Williamson's form if it can be written as in~\eqref{V xp}, with $Q=P=D$ diagonal. The diagonal entries of $D$, called symplectic eigenvalues of $V$, are uniquely determined by $V$ up to the order and no smaller than $1$ if $V$ is a QCM. They can be characterised as the positive eigenvalues of the matrix $i\Omega V$. Moreover, $D=\id$ if and only if the QCM is pure. This also shows that a $2m\times 2m$ QCM $V$ is pure if and only if
\bb
\rk (V+i\Omega) = m\, ,
\label{rank condition pure}
\ee
Note that upon taking the complex conjugate we can also rephrase this condition as $\rk (V-i\Omega)=m$. For more details, see the discussion below~\cite[Lemma~7]{LL-log-det}.

It is important to realise that if a QCM is in $xp$-form in the first place, it can be brought into Williamson's form by means of symplectics in the subgroup~\eqref{GL isomorphism}, as we argue below. Note that the subgroup in~\eqref{GL isomorphism} is also composed of matrices in $xp$-form!

\begin{lemma} \label{xp form Williamson lemma}
Let $V$ be a QCM in $xp$-form. Then there exists a symplectic matrix $S$ in $xp$-form (i.e.\ belonging to the subgroup~\eqref{GL isomorphism}) such that $SVS^\intercal$ is in Williamson's form.
\end{lemma}

\begin{proof}
Let $V$ be as in~\eqref{V xp}, and call $m$ the number of modes. Using the spectral theorem, choose an orthogonal $m\times m$ matrix $O$ with the property that $O^\intercal Q^{1/2} P Q^{1/2} O = D^2$ is diagonal. Set $M\coloneqq Q^{1/2} O D^{-1/2} $. We have that
\begin{align*}
M^{-1} Q M^{-\intercal} &= D^{1/2} O^\intercal Q^{-1/2}\, Q\, Q^{-1/2} O D^{1/2} = D\, ,\\
M^\intercal P M &= D^{-1/2} O^\intercal Q^{1/2}\, P\, Q^{1/2} O D^{-1/2} = D\, ,
\end{align*}
from which it follows that
\bbb
\begin{pmatrix} M^{-1} & 0 \\ 0 & M^\intercal \end{pmatrix} \begin{pmatrix} Q & 0 \\ 0 & P \end{pmatrix} \begin{pmatrix} M^{-\intercal} & 0 \\ 0 & M \end{pmatrix} = \begin{pmatrix} D & 0 \\ 0 & D \end{pmatrix}
\eee
is in Williamson's form.
\end{proof}

For more background on Williamson's  decomposition, we refer the reader to~\cite{GOSSON}. A particularly instructive case is that of pure QCMs: a pure QCM $\gamma$ is in $xp$-form if and only if
\bb
\gamma = \begin{pmatrix} Q & 0 \\ 0 & Q^{-1}  \end{pmatrix}
\label{pure in xp form}
\ee
with respect to an $xp$ block partition, with $Q>0$. To see that this is the case, it suffices to remember that pure QCMs are symmetric symplectic matrices, and to use the relation that defines symplecticity. Observe that pure QCMs in $xp$-form are exactly those symmetric matrices that belong to the subgroup~\eqref{GL isomorphism} of the symplectic group.

We can now give the following definition.

\begin{Def} \label{normal def}
A bipartite QCM $V_{AB}$ is called `normal' if it can be brought into $xp$-form by local symplectic operations, i.e.\ if there are local symplectic matrices $S_A$, $S_B$ such that
\bb
\Pi^x_{AB} \left(S_A\oplus S_B\right) V_{AB} \left( S_A^\intercal \oplus S_B^\intercal\right) (\Pi^p_{AB})^\intercal = 0\, .
\ee
\end{Def}

The following is an easy consequence of results of~\cite{giedkemode}.

\begin{lemma}[{\cite{giedkemode}}] \label{pure are normal}
All pure QCMs with arbitrary many modes on each side are normal.
\end{lemma}

The following result is due to Simon~\cite{Simon00} and independently to Duan et al.~\cite{Duan2000}.

\begin{lemma}[{\cite{Simon00, Duan2000}}] \label{2-mode are normal}
All two-mode QCMs are normal.
\end{lemma}

For a real parameter $t>0$, consider the QCMs on systems $A$ and $B$ defined by
\bb
\begin{aligned}
\Gamma_{\!A}(t) &\coloneqq t\, \big(\Pi^x_A\big)^{\!\intercal}\, \Pi^x_A + t^{-1} \big(\Pi^p_A\big)^{\!\intercal}\, \Pi^p_A = \begin{pmatrix} t\id_A^x & 0 \\ 0 & t^{-1}\id_A^p \end{pmatrix} \, ,\\
\Gamma_{\!B}(t) &\coloneqq t\, \big(\Pi^x_B\big)^{\!\intercal}\, \Pi^x_B + t^{-1} \big(\Pi^p_B\big)^{\!\intercal}\, \Pi^p_B = \begin{pmatrix} t\id_B^x & 0 \\ 0 & t^{-1}\id_B^p \end{pmatrix}\, ,
\end{aligned}
\label{seeds homodyne}
\ee
where the block matrices on the r.h.s.\ are written with respect to an $xp$ block partition. In the limit $t\to 0^+$, these QCMs are the seeds of two measurements of the $x$ quadrature, while in the opposite limit $t\to \infty$ they identify measurements of the $p$ quadrature. Occasionally, if there is no ambiguity on the partition used, we will remove superscripts and subscripts from expressions like~\eqref{seeds homodyne}.

\subsection{Properties of Schur complements} \label{Schur_subsec}

Here we discuss some applications of the notion of Schur complement. We have seen in the main text that for a $2\times 2$ block matrix
\bb
R=\begin{pmatrix} X & Z \\ Z^\intercal & Y \end{pmatrix}
\label{R}
\ee
one sets
\bb
R/X\coloneqq Y -Z^\intercal X^{-1} Z\, ,
\label{Schur}
\ee
provided that $X$ is invertible. The expression in~\eqref{Schur} is called the Schur complement of $R$ with respect to $X$. Schur complements are instrumental in matrix analysis~\cite{ZHANG}, and have found widespread applications in the context of CV quantum information as well~\cite{Lami16, LL-log-det}. Here we limit ourselves to recalling a few notable properties of these objects, referring the interested reader to~\cite{ZHANG}.
\begin{enumerate}
\item[(i)] Schur determinant factorisation formula:
\bb
\det R = (\det X) \big(\det (R/X)\big)
\label{det fact}
\ee
\item[(ii)] Inertia additivity formula:
\bb
\rk R = \rk X + \rk (R/X)\, .
\label{inertia additivity}
\ee
\item[(iii)] Block inverse:
\bb
R^{-1} = \begin{pmatrix} (R/Y)^{-1} & & -X^{-1} Y (R/X)^{-1} \\[0.7ex] -(R/X)^{-1} Y^T X^{-1} && (R/X)^{-1} \end{pmatrix} ,
\label{inv}
\ee
provided that $R$ as well as $X$ are invertible.
\item[(iv)] Block positivity conditions:
\bb
R>0\qquad \Longleftrightarrow\qquad X>0\quad \text{and}\quad R/X>0\, .
\label{Schur block positivity}
\ee
\item[(v)] Covariance under congruences: for all invertible $M,N$, it holds that
\bb
\left((M\oplus N)\, R\, (M\oplus N)^\intercal \right) \big/ (MXM^\intercal) = N \left(R/X\right) N^\intercal\, .
\label{Schur covariance}
\ee
\item[(vi)] Variational representation: when $R>0$ is positive definite, it holds that
\bb
R/X = \max\left\{ T:\, R\geq 0\oplus T \right\} .
\label{Schur variational}
\ee
Note that the set on the right-hand side is a matrix set, hence it is not a priori guaranteed to have a maximum. What~\eqref{Schur variational} tells us is that such a maximum however does exist, and that it coincides with the Schur complement of $R$ with respect to $X$. From~\eqref{Schur variational} it is straightforward to deduce the following monotonicity property: if $R>0$ and $R'=\lsmatrix X' & Z' \\ (Z')^\intercal & Y' \rsmatrix>0$, then
\bb
R\geq R' \quad \Longrightarrow\quad R/X \geq R'/X'\, .
\label{Schur monotonicity}
\ee
\end{enumerate}


\subsection{Properties of the log-determinant mutual information} \label{sec properties IM}

Here we summarise known properties of the log-determinant mutual information~\eqref{IM} and establish some new ones --- such as a type of uniform continuity --- that will play an important role in the proof of Theorem~\ref{equality normal thm}. A list of basic properties is as follows.
\begin{enumerate}[(i)]
\item Invariance under local symplectics: for all pairs of symplectic matrices $S_A, S_B$, it holds that
\bb
I_M(A:B)_{(S_A\oplus S_B)\, V_{AB}\, (S_A^\intercal \oplus S_B^\intercal)} \equiv I_M(A:B)_{V_{AB}}\, .
\label{IM invariance symplectics}
\ee
\item Invariance under rescaling: whenever $t>0$, we have that
\bb
I_M(A:B)_{V} \equiv I_M(A:B)_{tV} \, .
\label{IM invariance rescaling}
\ee
\item Data processing inequality: for every positive semidefinite matrix $K_A\geq 0$ on system $A$, it holds that
\bb
I_M(A:B)_{V_{AB} + K_A} \leq I_M(A:B)_{V_{AB}}\, ,
\label{data processing}
\ee
where $K_A$ is a shorthand for $K_A\oplus 0_B$. This is clear if one thinks of $I_M(A\!:\!B)_V$ as the Shannon mutual information of a Gaussian random variable with covariance matrix $V$. Adding a local positive semidefinite matrix corresponds to adding an independent and normal local random variable.
\item For every pure QCM $\gamma_{AB}$, it holds that
\bb
I_M(A:B)_\gamma = 2M(\gamma_A)=2M(\gamma_B)\, .
\label{IM pure}
\ee
This follows trivially from the fact that the local reduced states corresponding to the pure Gaussian state with QCM $\gamma_{AB}$ have the same R\'enyi-$2$ entropies, and moreover $M(\gamma_{AB})=0$ as $\det \gamma_{AB}=1$.
\item For all positive matrices $V>0$, it holds that~\cite[Eq.~(29)]{LL-log-det}
\bb
I_M(A:B)_V = I_M(A:B)_{V^{-1}}\, .
\label{IM inversion}
\ee
\end{enumerate}

We now explore some further properties of the log-determinant mutual information~\eqref{IM}. Consider a bipartite QCM $V_{AB}$, and define
\begin{align}
V_{AB}^x &\coloneqq \Pi_{AB}^x V_{AB} (\Pi_{AB}^x)^\intercal\, , \label{x_proj} \\
V_{AB}^p &\coloneqq \Pi_{AB}^p V_{AB} (\Pi_{AB}^p)^\intercal\, . \label{p_proj}
\end{align}
Since $V_{AB}^x$ and $V_{AB}^p$ are principal submatrices of $V_{AB}$, we can take Schur complements with respect to them. Also, since they still retain formally a block form with respect to the partition $A:B$, we can also compute their log-determinant mutual information, which we denote for instance by $I_M(A_x\!:\!B_x)_{V_{AB}^x}$. With this notation in mind, we start by showing the following.

\begin{lemma} \label{I M decomposition xp lemma}
For all bipartite QCMs $V_{AB}$, the log-determinant mutual information admits the decomposition
\bb
I_M(A:B)_{V_{AB}} = I_M(A_x:B_x)_{V_{AB}^x} + I_M(A_p:B_p)_{V_{AB}/V_{AB}^x}\, .
\label{I M decomposition xp}
\ee
\end{lemma}

\begin{proof}
It suffices to apply repeatedly the Schur determinant factorisation formula~\eqref{det fact}
\begin{align*}
I_M(A:B)_{V_{AB}}\! &= \frac12 \log_2 \frac{(\det V_A) (\det V_B)}{\det V_{AB}} \\
&= \frac12 \log_2 \frac{(\det V_{A}^x) (\det V_A/V_A^x) (\det V_B^x) (\det V_B/V_B^x)}{(\det V_{AB}^x) (\det V_{AB}/V_{AB}^x)} \\
&= \frac12 \log_2 \frac{(\det V_{A}^x) (\det V_B^x) }{(\det V_{AB}^x)} \\
&+ \frac12 \log_2 \frac{(\det V_A/V_A^x) (\det V_B/V_B^x)}{ (\det V_{AB}/V_{AB}^x)} \\
&= I_M(A_x:B_x)_{V_{AB}^x} + I_M(A_p:B_p)_{V_{AB}/V_{AB}^x}\, .
\end{align*}
This concludes the proof.
\end{proof}

\begin{cor} \label{pure_xp_IM_cor}
Let $\gamma_{AB}$ be a pure QCM in $xp$-form, i.e.\ let it be as in~\eqref{pure in xp form}. Then
\bb
I_M(A_x:B_x)_{\gamma_{AB}^x} = I_M(A_p:B_p)_{\gamma_{AB}^p} = \frac12 I_M(A:B)_{\gamma_{AB}} = M(\gamma_A) \, .
\label{pure_xp_IM}
\ee
\end{cor}

\begin{proof}
Remembering~\eqref{pure in xp form} and~\eqref{IM inversion}, we see that
\bbb
I_M(A_x:B_x)_{\gamma_{AB}^x} = I_M(A_p:B_p)_{\left(\gamma_{AB}^p\right)^{-1}} = I_M(A_p:B_p)_{\gamma_{AB}^p}\, .
\eee
Also, since $\gamma_{AB}$ has no off-diagonal terms, it holds that $\gamma_{AB}/\gamma_{AB}^x = \gamma_{AB}^p$. Combining this with~\eqref{I M decomposition xp} and~\eqref{IM pure} yields
\begin{align*}
2M(\gamma_A) &= I_M(A:B)_{\gamma_{AB}} = I_M(A_x:B_x)_{\gamma_{AB}^x} + I_M(A_p:B_p)_{\gamma_{AB}^p} \\
&= 2 I_M(A_x:B_x)_{\gamma_{AB}^x}\, ,
\end{align*}
completing the proof.
\end{proof}

\begin{lemma}[(Uniform continuity of log-determinant mutual information)] \label{continuity mutual information lemma}
Let $\kappa\geq 1$, and consider two matrices $V_{AB},W_{AB}$ such that $V_{AB},W_{AB}\geq \kappa^{-1} \id_{AB}$. Then it holds that
\bb
\left| I_M(A:B)_V - I_M(A:B)_W\right| \leq \kappa \log_2(e)\,\|V_{AB} - W_{AB}\|_1\, ,
\label{continuity mutual information}
\ee
where $\|\cdot\|_1$ denotes the trace norm.
\end{lemma}

\begin{proof}
Let us start by proving that for any two matrices $M,N\geq \kappa^{-1} \id$, it holds that
\bb
\left|\log_2 \det (MN^{-1})\right| \leq \kappa \log_2(e)\, \|M-N\|_1\, .
\label{inequality M,N}
\ee
To see this, we write
\begin{align*}
\left|\log_2 \det (MN^{-1})\right| &= \left|\log_2 \det (M) - \log_2 \det (N)\right| \\
&\eqt{1} \left|\sumno_i \left(\log_2 \mu_i - \log_2 \nu_i\right)\right| \\
&\leq \sumno_i \left|\log_2 \mu_i - \log_2 \nu_i\right| \\
&\leqt{2} \kappa \log_2(e)\, \sumno_i |\mu_i - \nu_i| \\
&\leqt{3} \kappa \log_2(e)\, \|M-N\|_1\, .
\end{align*}
Here, in~1 we called $\mu_i,\nu_i$ the eigenvalues of $M,N$, sorted in descending order. To justify~2, instead, start by applying Lagrange's theorem to a continuously differentiable function $f:[a,b]\to \R$:
\bbb
\frac{\left|f(\mu)-f(\nu)\right|}{|\mu-\nu|} \leq \max_{a\leq \xi\leq b} \left| f'(\xi)\right| .
\eee
Setting $f(x)= \log_2 x$, $a=\kappa^{-1}$, and $b=\infty$, we obtain the desired inequality $\left|\log_2 \mu_i - \log_2 \nu_i\right| \leq \kappa \log_2(e)\, |\mu_i - \nu_i|$. Finally, in~3 we used the well-known estimate $\|M-N\|_1\geq \sumno_i |\mu_i-\nu_i|$, which is due to Mirsky~\cite[Theorem~5]{Mirsky1960} (see also~\cite[Eq.~(IV.62)]{BHATIA-MATRIX}).

We now come to the proof of~\eqref{continuity mutual information}:
\begin{align*}
&\left| I_M(A:B)_V - I_M(A:B)_W\right| \\&= \left| \frac12 \log_2 \frac{\det(V_A) \det (V_B)}{\det (V_{AB})} - \frac12 \log_2 \frac{\det(W_A) \det (W_B)}{\det (W_{AB})} \right| \\[1ex]
&= \bigg| \frac12 \log_2 \det\left(V_AW_A^{-1}\right) + \frac12 \log_2 \det \left(V_BW_B^{-1}\right) \\
& \qquad- \frac12 \log_2 \det \left(V_{AB}W_{AB}^{-1}\right) \bigg| \\[1ex]
&\leqt{3} \frac{\kappa \log_2(e)}{2} \left( \|V_A-W_A\|_1 + \|V_B-W_B\|_1 + \|V_{AB} - W_{AB}\|_1\right) \\[1ex]
&\leqt{4} \kappa \log_2(e)\, \left\|V_{AB} - W_{AB}\right\|_1\, .
\end{align*}
Here,~3 comes from inequality~\eqref{inequality M,N} applied to $M=V_A, V_B, V_{AB}$ and $N=W_A,W_B, W_{AB}$, respectively. Note that under the operation of taking principal submatrices the minimal eigenvalue never decreases, hence $V_A, W_A\geq \kappa^{-1}\id_{A}$ and $V_B, W_B\geq \kappa^{-1} \id_{B}$. Finally,~4 descends from the observation that since the `pinching' operation that maps $V_{AB}\mapsto V_A\oplus V_B$ and $W_{AB}\mapsto W_A\oplus W_B$ is a quantum channel, it never increases the trace norm, and hence
\begin{align*}
\left\| V_{AB} - W_{AB}\right\|_1 &\geq \left\| V_A\oplus V_B - W_A\oplus W_B\right\|_1 \\
&= \| V_A-W_A\|_1 + \| V_B - W_B\|_1\, .
\end{align*}
This concludes the proof.
\end{proof}


For the sake of completeness, here we further present a simple justification of the known fact~\cite{giedkemode, Lada2011, GIE-PRA} that in the Gaussian setting the classical mutual information of Gaussian pure states coincides with the local R\'enyi-$2$ entropy, as claimed in~\eqref{IMc pure states}. 

\begin{lemma} \label{classical_mutual_info_pure_lemma}
For all pure QCMs $\gamma_{AB}$, it holds that $I^c_M(A:B)_\gamma = \frac12 I_M(A:B)_\gamma = M(\gamma_A)$. Moreover, if $\gamma_{AB}$ is in $xp$-form, the optimal Gaussian measurements in~\eqref{IMc} are identical homodynes.
\end{lemma}

\begin{proof}
Start by employing~\eqref{data processing} to deduce that
\bbb
I_M(A:B)_{\gamma_{AB} + \Gamma_A\oplus \Gamma_B}\leq I_M(A:B)_{\gamma_{AB} + \Gamma_B}\, ,
\eee
for all $\Gamma_A\geq 0$. Using the determinant factorisation formula~\eqref{det fact}, it is not difficult to show that
\bbb
I_M(A:B)_{\gamma_{AB} + \Gamma_B} = M(\gamma_A) - M\left( (\gamma_{AB} + \Gamma_B) / (\gamma_B + \Gamma_B) \right) \leq M(\gamma_A)\, ,
\eee
where the last inequality follows because $(\gamma_{AB} + \Gamma_B) / (\gamma_B + \Gamma_B)$ is a QCM by~\eqref{elementary_pure_G_lemma}, and hence the corresponding log-determinant entropy must be non-negative. Consequently, one obtains that $I_M(A:B)_{\gamma_{AB} + \Gamma_A \oplus \Gamma_B} \leq M(\gamma_A)$. Since $\Gamma_A, \Gamma_B$ were arbitrary, we conclude that
\bb
I_M^c(A:B)_{\gamma} \leq M(\gamma_A)\, ,
\ee
To prove the converse inequality, we can use the normality of pure QCMs (Lemma~\ref{pure are normal}) to find local symplectics $S_A,S_B$ such that
\bb
\gamma_{AB} = (S_A\oplus S_B)\, \gamma_{AB}^{(0)}\, (S_A\oplus S_B)^\intercal\, ,
\ee
where $\gamma_{AB}^{(0)}$ is in $xp$-form. Observe that because of~\eqref{pure in xp form} we have that $\gamma_{AB}^{(0)} = \lsmatrix Q & 0 \\ 0 & Q^{-1} \rsmatrix$ with respect to an $xp$ block partition, where $Q>0$. Obviously, if $\gamma$ is in $xp$-form we can set $S_A=\id_A$ and $S_B=\id_B$ in the first place. Now, construct
\begin{align*}
\widetilde{\Gamma}_A(t) &= S_{A} \Gamma_A(t) S_{A}^\intercal\, ,\\
\widetilde{\Gamma}_B(t) &= S_{B} \Gamma_B(t) S_{B}^\intercal\, ,
\end{align*}
where $\Gamma_A(t), \Gamma_B(t)$ are the seeds of a homodyne as given by~\eqref{seeds homodyne}. With respect to an $xp$ block partition, one has that
\bbb
\gamma_{AB}^{(0)} = \begin{pmatrix} Q+ t\id & 0 \\ 0 & Q^{-1}+t^{-1}\id \end{pmatrix} ,
\eee
and hence that
\begin{align*}
&\lim_{t\rightarrow 0} I_M(A:B)_{\gamma_{AB} + \widetilde{\Gamma}_A (t) \oplus \widetilde{\Gamma}_B (t)}\\ &\eqt{1} \lim_{t\rightarrow 0} I_M(A:B)_{\gamma_{AB}^{(0)} + \Gamma_A (t) \oplus \Gamma_B (t)} \\
&\eqt{2} \lim_{t\rightarrow 0} \left( I_M(A_x:B_x)_{Q+t\id} + I_M(A_p:B_p)_{Q^{-1}+t^{-1}\id} \right) \\
&\eqt{3} \lim_{t\rightarrow 0} \left( I_M(A_x:B_x)_{Q+t\id} + I_M(A_p:B_p)_{\id+tQ^{-1}} \right) \\
&\eqt{4} I_M(A_x:B_x)_Q \\
&\eqt{5} M\big(\gamma^{(0)}_A\big) \\
&= M(\gamma_A)\, .
\end{align*}
The above identities can be justified as follows: 1~descends from~\eqref{IM invariance symplectics}; 2~is an application of~\eqref{I M decomposition xp}; 3~uses the identity~\eqref{IM invariance rescaling}; 4~comes from Lemma~\ref{continuity mutual information lemma}, which implies that
\begin{align*}
I_M(A_p:B_p)_{\id + t Q^{-1}} &= \left| I_M(A_p:B_p)_{\id + t Q^{-1}}- I_M(A:B)_\id\right|
\\ &\leq t \log_2(e)\, \|Q^{-1}\|_1 \tends{}{t\to 0} 0
\end{align*}
and that
\begin{align*}
&\left| I_M(A_p:B_p)_{Q + t \id} - I_M(A_p:B_p)_Q\right| \\
&\qquad \leq t \log_2(e)\, \|\id_p\|_1 \|Q^{-1}\|_1 \tends{}{t\to 0} 0\, ;
\end{align*}
finally, 5~is a consequence of Corollary~\ref{pure_xp_IM_cor}.
\end{proof}

\subsection{A few more technical lemmata}\label{sec mammeta}

Here we need to establish a couple of technical lemmata which will be useful for the proof of Theorem~\ref{equality normal thm}.

\begin{lemma} \label{measured QCMs compact}
Let $V_{AB}$ be any bipartite QCM, and let $\gamma_{ABE}$ be any extension of it, i.e.\ let it satisfy $\gamma_{AB}=V_{AB}$. For some QCM $\Gamma_E$, define
\bb
W_{AB} \coloneqq \left( \gamma_{ABE} + \Gamma_E\right) \big/ \left(\gamma_{E} +\Gamma_E\right) .
\ee
Then
\bb
\lambda_{\max}(V_{AB})^{-1} \id_{AB} \leq \Omega_{AB}^\intercal V_{AB}^{-1} \Omega_{AB} \leq W_{AB}\leq V_{AB} \, ,
\ee
where $\lambda_{\max}(V_{AB})$ denotes the maximal eigenvalue of $V_{AB}$.
\end{lemma}

\begin{proof}
The very definition of Schur complement implies that $W_{AB}\leq \gamma_{AB} = V_{AB}$. As for the opposite inequality, since $\Gamma_{\!E}\geq 0$ is positive semidefinite and the Schur complement is monotonic~\eqref{Schur monotonicity}, one has that
$\left( \gamma_{ABE} + \Gamma_{\!E}\right) \big/ \left(\gamma_{E} +\Gamma_{\!E}\right) \geq \gamma_{ABE}/\gamma_{E} \geq \Omega_{AB}^\intercal \gamma_{AB}^{-1}\Omega_{AB} = \Omega_{AB}^\intercal V_{AB}^{-1}\Omega_{AB}$,
where the second inequality is an application of~\cite[Theorem~3]{Lami16}.
\end{proof}

\begin{prop} \label{homodyne limit prop}
Let $\Gamma_{\!A}(t)$ and $\Gamma_{\!B}(t)$ be the seeds of homodyne measurements defined by~\eqref{seeds homodyne}. For some $0<\kappa\leq 1$, let $\pazocal{W}_{AB} \geq \kappa^{-1} \id_{AB}$ be a compact set of bipartite QCMs that is bounded away from zero \footnote{Here we mean that every element $W_{AB}\in \pazocal{W}_{AB}$ satisfies $W_{AB}\geq \kappa^{-1}\id_{AB}$}. Then
\bb
\lim_{t\to 0^+} \inf_{W_{AB}\in \pazocal{W}_{AB}} I_M(A:B)_{W_{AB} + \Gamma_A(t)\oplus \Gamma_B(t)} = \inf_{W_{AB}\in \pazocal{W}_{AB}} I_M(A_x:B_x)_{W^x_{AB}}\, ,
\label{homodyne limit}
\ee
where $W^x_{AB}$ is defined analogously to~\eqref{x_proj}.
\end{prop}

\begin{proof}
Let us define the functions $F_t:\pazocal{W}_{AB}\to \R$ and $F:\pazocal{W}_{AB}\to \R$ via
\begin{align}
F_t(W_{AB}) &\coloneqq I_M(A\!:\!B)_{W_{AB} + \Gamma_{\!A}(t)\oplus \Gamma_{\!B}(t)}\, , \label{F t} \\
F(W_{AB}) &\coloneqq I_M(A_x\!:\!B_x)_{W^x_{AB}}\, . \label{F}
\end{align}
Then~\eqref{homodyne limit} becomes
\bb
\lim_{t\to 0^+} \inf_{W_{AB}\in \pazocal{W}_{AB}} F_t(W_{AB}) = \inf_{W_{AB}\in \pazocal{W}_{AB}} F(W_{AB})\, .
\label{homodyne limit F}
\ee
To prove~\eqref{homodyne limit F}, it suffices to show that $F_t \tends{u}{t\to 0^+} F$ \emph{uniformly}, i.e.\ that
\bb
\lim_{t\to 0^+} \|F_t - F\|_\infty = 0\, ,
\label{uniform convergence F t}
\ee
where
\begin{align}
\|F_t - F\|_\infty \coloneqq&\ \sup_{W_{AB}\in\pazocal{W}_{AB}} \left| F_t(W_{AB}) - F(W_{AB}) \right| \label{sup norm W 1}
\\
=&\ \sup_{W_{AB}\in\pazocal{W}_{AB}} \left| I_M(A\!:\!B)_{W_{AB} + \Gamma_{\!A}(t)\oplus \Gamma_{\!B}(t)} - I_M(A_x\!:\!B_x)_{W^x} \right| . \label{sup norm W 2}
\end{align}
The fact that~\eqref{uniform convergence F t} implies~\eqref{homodyne limit F} is well known, see for instance~\cite[Chapter~7]{RUDIN-PRINCIPLES}. 
We thus proceed to prove~\eqref{uniform convergence F t}. Note that we are not in the position to apply Lemma~\ref{continuity mutual information lemma} on the continuity of mutual information, for example because the QCMs $W_{AB} + \Gamma_{\!A}(t)\oplus \Gamma_{\!B}(t)$ do not converge to anything as $t\to 0^+$, due to the presence of the diverging term $t^{-1}$ in~\eqref{seeds homodyne}. To proceed further, let us write down an explicit decomposition of an arbitrary matrix $W_{AB}\in\pazocal{W}_{AB}$ with respect to an $xp$ block partition as
\bb
W_{AB} = \begin{pmatrix} W_{AB}^x & Z_{AB} \\ Z_{AB}^\intercal & W_{AB}^p \end{pmatrix} .
\label{W AB blocks}
\ee
Observe that according to the same splitting one has that
\bb
\Gamma_A(t) \oplus \Gamma_B(t) = \begin{pmatrix} t\, \id_{AB}^x & 0 \\ 0 & t^{-1} \id_{AB}^p \end{pmatrix} .
\label{seeds homodyne blocks}
\ee
We now massage~\eqref{F t} as follows:
\begin{align*}
F_t(W_{AB}) &= I_M(A:B)_{W_{AB} + \Gamma_A(t)\oplus \Gamma_B(t)} \\
&\eqt{1} I_M(A_x:B_x)_{W_{AB}^x + \Gamma_{\!A}^x(t)\oplus \Gamma_B^x(t)} \\ &\qquad + I_M(A_p:B_p)_{\big(W_{AB} + \Gamma_A(t)\oplus \Gamma_B(t)\big)\big/ \big(W_{AB}^{x} + \Gamma_A^x(t)\oplus \Gamma_B^x(t)\big)} \\
&\eqt{2} I_M(A_x:B_x)_{W_{AB}^x + t\, \id_{AB}^x} \\ &\qquad + I_M(A_p:B_p)_{W_{AB}^p + t^{-1} \id_{AB}^p - Z_{AB}^\intercal \left(W_{AB}^x + t\, \id_{AB}^x\right)^{-1} Z_{AB}} \\
&\eqt{3} I_M(A_x:B_x)_{W_{AB}^x + t\, \id_{AB}^x} \\ &\qquad + I_M(A_p:B_p)_{\id_{AB}^p + t\left( W_{AB}^p  - Z_{AB}^\intercal \left(W_{AB}^x + t\, \id_{AB}^x\right)^{-1} Z_{AB}\right)}
\end{align*}
Note that in~1 we applied Lemma~\ref{I M decomposition xp lemma}, in~2 we made use of~\eqref{W AB blocks} and~\eqref{seeds homodyne blocks}, and in~3 we employed~\eqref{IM invariance rescaling}.
Now, we see that
\begin{align*}
&\left| F_t(W_{AB}) - F(W_{AB})\right| \\ &\leqt{4} \left| I_M(A_x:B_x)_{W_{AB}^x + t\, \id_{AB}^x} - I_M(A_x:B_x)_{W_{AB}^x} \right| \\
&\quad\! + \left| I_M(A_p:B_p)_{\id_{AB}^p + t\left( W_{AB}^p  - Z_{AB}^\intercal \left(W_{AB}^x + t\, \id_{AB}^x\right)^{-1} Z_{AB}\right)} - I_M(A_p:B_p)_{\id_{AB}^p}\right| \\
&\leqt{5} \kappa \log_2(e) \left\|W_{AB}^x + t\, \id_{AB}^x - W_{AB}^x \right\|_1 \\
&\quad\! + \log_2(e) \left\| \id_{AB}^p + t\left( W_{AB}^p  - Z_{AB}^\intercal \left(W_{AB}^x + t\, \id_{AB}^x\right)^{-1} Z_{AB}\right) - \id_{AB}^p\right\|_1 \\
&\eqt{6} t \log_2(e) \left(\kappa (m_A+m_B) + \max_{W_{AB} \in \pazocal{W}_{AB}} \Tr \left[W_{AB}^p\right]\right) .
\end{align*}
Note that the inequality in~4 comes directly from the definitions~\eqref{F t} and~\eqref{F}, together with the observation that $I_M(A_p\!:\!B_p)_{\id_{AB}^p}=0$. For~5, instead, we applied Lemma~\ref{continuity mutual information lemma} twice, which is possible because: (i)~$W_{AB}\geq \kappa^{-1}\id_{AB}$ and hence $W_{AB}^x\geq \kappa^{-1}\id_{AB}^x$ (also, $t>0$); and (ii)~from the block positivity conditions~\eqref{Schur block positivity} for $W_{AB}+t\, \id_{AB}$ it follows that
$W_{AB}^p  - Z_{AB}^\intercal \left(W_{AB}^x + t\, \id_{AB}^x\right)^{-1} Z_{AB} = (W_{AB} + t\,\id_{AB}^x) \big/ (W_{AB}^x + t\, \id_{AB}^x) > 0$.
To justify~6, note that exploiting again positivity one has that
$\left\| (W_{AB} + t\,\id_{AB}^x) \big/ (W_{AB}^x + t\, \id_{AB}^x)\right\|_1 = \Tr\left[(W_{AB} + t\,\id_{AB}^x) \big/ (W_{AB}^x + t\, \id_{AB}^x)\right] \leq \Tr[W_{AB}^p]$.
The quantity $\max_{W_{AB} \in \pazocal{W}_{AB}} \Tr \left[W_{AB}^p\right]$ is finite thanks to the compactness of $\pazocal{W}_{AB}$. Also, we denoted by $m_A$ and $m_B$ the number of modes on systems $A$ and $B$, respectively. From the above estimate it is clear that $\lim_{t\to 0^+} \sup_{W_{AB} \in \pazocal{W}_{AB}} \left| F_t(W_{AB}) - F(W_{AB})\right| = 0$, completing the argument.
\end{proof}

\section{A non-normal quantum covariance matrix} \label{sec non-normal}

Here we construct an explicit example of a non-normal QCM over a $(1+2)$-mode system. Let us start by recalling a well-known result in symplectic geometry.

\begin{lemma}[{\cite[Proposition~8.12]{GOSSON}}] \label{Gosson lemma}
Let $A>0$ be a $2n\times 2n$ positive definite matrix. Let $S_1,S_2$ be two symplectic matrices that bring $A$ into Williamson's form, i.e.\ let them satisfy
\bb
A = S_1 \begin{pmatrix} D & 0 \\ 0 & D \end{pmatrix} S_1^\intercal = S_2 \begin{pmatrix} D & 0 \\ 0 & D \end{pmatrix} S_2^\intercal \, ,
\ee
where $D = \sum_{j=1}^n d_j \ketbra{j} > 0$ is diagonal. Then
\bb
S_2=S_1K\, , \qquad K = H \begin{pmatrix} U & 0 \\ 0 & U^* \end{pmatrix} H^\dag \in \mathbb{K}(n)\, ,\qquad [U,D]=0\, .
\ee
The last condition means that $U$ only mixes vectors $\ket{j}$ with the same coefficients $d_j$. If these are all distinct, then necessarily $U=\sum_j e^{i\theta_j} \ketbra{j}$ for some phases $\theta_j\in\R$.
\end{lemma}

\begin{cor} \label{xp Williamson cor}
Let $D>0$ be an $n\times n$ diagonal matrix. A $2n\times 2n$ symplectic matrix $S$ is such that $S\left( \begin{smallmatrix} D & 0 \\ 0 & D \end{smallmatrix} \right) S^\intercal$ is in xp-form, i.e.
\bb
S \begin{pmatrix} D & 0 \\ 0 & D \end{pmatrix} S^\intercal = \begin{pmatrix} X & 0 \\ 0 & P \end{pmatrix} ,
\ee
if and only if
\bb
S = \begin{pmatrix} M^{-1} & 0 \\ 0 & M^\intercal  \end{pmatrix} H \begin{pmatrix} U & 0 \\ 0 & U^* \end{pmatrix} H^\dag ,
\label{xp Williamson}
\ee
where $M$ is an arbitrary $n\times n$ invertible matrix, $H$ is given by~\eqref{H}, and $[U,D]=0$.
\end{cor}

\begin{proof}
By Lemma~\ref{xp form Williamson lemma}, we can find an invertible matrix $M$ such that
\bbb
\begin{pmatrix} M & 0 \\ 0 & M^{-\intercal} \end{pmatrix} S \begin{pmatrix} D & 0 \\ 0 & D \end{pmatrix} S^\intercal \begin{pmatrix} M^{\intercal} & 0 \\ 0 & M^{-1} \end{pmatrix} = \begin{pmatrix} \widetilde{D} & 0 \\ 0 & \widetilde{D} \end{pmatrix} ,
\eee
where $M^{-\intercal}\coloneqq \big(M^{-1}\big)^\intercal$, and $\left( \begin{smallmatrix} M & 0 \\ 0 & M^{-\intercal} \end{smallmatrix} \right)$ is symplectic by~\eqref{GL isomorphism}. Thanks to the uniqueness of symplectic eigenvalues, we must have that $\widetilde{D}=D$, up to an immaterial permutation that can always be absorbed into $M$. We can then apply Lemma~\ref{Gosson lemma}, guaranteeing that
\bbb
\begin{pmatrix} M & 0 \\ 0 & M^{-\intercal} \end{pmatrix} S = H \begin{pmatrix} U & 0 \\ 0 & U^* \end{pmatrix} H^\dag\, ,
\eee
with $[U,D]=0$. Multiplying by $\left( \begin{smallmatrix} M^{-1} & 0 \\ 0 & M^{\intercal} \end{smallmatrix} \right)$ from the left proves~\eqref{xp Williamson}.
\end{proof}

\begin{prop}[(A family of non-normal QCMs)]
Consider a bipartite $3$-mode system composed by modes $A$ on one side and $B_1,B_2$ on the other. Construct the QCM
\bb
V_{AB}\coloneqq \begin{pmatrix} a\id & F & G \\ F^\intercal & b_1\id & 0 \\ G^\intercal & 0 & b_2\id \end{pmatrix} ,
\ee
where the first and second rows and columns refer to $A$, the third and fourth to $B_1$, and the fifth and sixth to $B_2$. Assume that $\|F\|_\infty,\|G\|_\infty\leq 1$, $a\geq 3$, and $b_1,b_2\geq 2$, so that $V_{AB}\geq \id_{AB}$ is a valid QCM. Further assume that $b_1\neq b_2$ and that $\left[ FF^\intercal,GG^\intercal \right]\neq 0$. Then $V_{AB}$ is not normal.
\end{prop}

\begin{proof}
Assume by contradiction that
$(S_A\oplus S_B)\, V_{AB}\, (S_A\oplus S_B)^\intercal$ is in $xp$-form. Then there are scalars $x,p>0$ and $2\times 2$ matrices $X',P'>0$ such that
\bbb
S_A \begin{pmatrix} a & 0 \\ 0 & a \end{pmatrix} S_A^\intercal = \begin{pmatrix} x & 0 \\ 0 & p \end{pmatrix},\qquad S_B \begin{pmatrix} b_1 & 0 & 0 & 0 \\ 0 & b_2 & 0 & 0 \\ 0 & 0 & b_1 & 0 \\ 0 & 0 & 0 & b_2 \end{pmatrix} S_B^\intercal = \begin{pmatrix} X' & 0 \\ 0 & P' \end{pmatrix} .
\eee
If $b_1\neq b_2$, any unitary $U$ satisfying $\left[ U, \left( \begin{smallmatrix} b_1 & 0 \\ 0 & b_2 \end{smallmatrix}\right) \right]=0$ must be of the form $U = \left( \begin{smallmatrix} e^{i\theta_1} & 0 \\ 0 & e^{i\theta_2}\end{smallmatrix} \right)$. Hence, Corollary~\ref{xp Williamson cor} tells us that
\begin{align*}
S_A &= \begin{pmatrix} m^{-1} & 0 \\ 0 & m \end{pmatrix} H \begin{pmatrix} e^{i\varphi} & 0 \\ 0 & e^{-i\varphi} \end{pmatrix} H^\dag \\[1ex]
&= \begin{pmatrix} m^{-1} & 0 \\ 0 & m \end{pmatrix} \begin{pmatrix} \cos \varphi & \sin \varphi \\ -\sin\varphi & \cos\varphi \end{pmatrix} \\[1ex]
&= \begin{pmatrix} m^{-1} & 0 \\ 0 & m \end{pmatrix} R(\varphi)
\end{align*}
and
\begin{align*}
S_B &= \begin{pmatrix} N^{-1} & 0 \\ 0 & N^\intercal \end{pmatrix} H \begin{pmatrix} e^{i\theta_1} & 0 & 0 & 0 \\ 0 & e^{i\theta_2} & 0 & 0 \\ 0 & 0 & e^{-i\theta_1} & 0 \\ 0 & 0 & 0 & e^{-i\theta_2} \end{pmatrix} H^\dag \\[1ex]
&= \begin{pmatrix} N^{-1} & 0 \\ 0 & N^\intercal \end{pmatrix} \begin{pmatrix} \cos \theta_1 & 0 & \sin\theta_1 & 0 \\ 0 & \cos \theta_2 & 0 & \sin\theta_2 \\ -\sin\theta_1 & 0 & \cos\theta_1 & 0 \\ 0 & -\sin \theta_2 & 0 & \cos\theta_2 \end{pmatrix} \\[1ex]
&= \begin{pmatrix} N^{-1} & 0 \\ 0 & N^\intercal \end{pmatrix} (R_{B_1}(\theta_1) \oplus R_{B_2}(\theta_2))\, ,
\end{align*}
where $\varphi,\theta_1,\theta_2\in \R$, $m\neq 0$, and $N$ is $2\times 2$ and invertible. In order for the $AB_1$ and $AB_2$ off-diagonal blocks to be brought into $xp$-form as well, we must also request that
\bbb
R(\varphi) F R_{B_1}(\theta_1)^\intercal = D_1\, ,\qquad R(\varphi) G R_{B_2}(\theta_2)^\intercal = D_2
\eee
be both diagonal. This would imply that
\begin{align*}
0 &= \left[D_1^2, D_2^2\right] \\
&= \left[ R(\varphi)FF^\intercal R(\varphi)^\intercal,\, R(\varphi)GG^\intercal R(\varphi)^\intercal\right] \\
&= R(\varphi) \left[ FF^\intercal,\, GG^\intercal \right] R(\varphi)^\intercal\, ,
\end{align*}
so that $\left[ FF^\intercal,\, GG^\intercal \right]=0$, contrary to the hypotheses.
\end{proof}


\end{document}